\documentclass[11pt]{article}

\usepackage{natbib}
\usepackage{amsmath}
\usepackage{amsthm}
\usepackage{algorithmic}
\usepackage{ bbold }
\usepackage{multirow}
\usepackage{graphicx}
\usepackage{aliascnt}

\DeclareMathAlphabet{\mathpzc}{OT1}{pzc}{m}{it}

\newcommand{\agind}[1][i]{_{#1}}

\newcommand{\ironed}{\bar}
\newcommand{\constrained}{\hat}
\newcommand{\optconstrained}{\composed{\optimized}{\constrained}}
\newcommand{\optimized}[1]{#1\opt}
\newcommand{\differentiated}[1]{#1'}

\newcommand{\primedarg}[1]{#1\primed}
\newcommand{\noaccents}[1]{#1}
\newcommand{\composed}[3]{#1{#2{#3}}}

\newcommand{\prob}[2][]{\text{\bf Pr}\ifthenelse{\not\equal{}{#1}}{_{#1}}{}\!\left[#2\right]}
\newcommand{\expect}[2][]{\text{\bf E}\ifthenelse{\not\equal{}{#1}}{_{#1}}{}\!\left[#2\right]}
\newcommand{\given}{\,\mid\,}

\newcommand{\tparen}{\big}
\newcommand{\tprob}[2][]{\text{\bf Pr}\ifthenelse{\not\equal{}{#1}}{_{#1}}{}\tparen[#2\tparen]}
\newcommand{\texpect}[2][]{\text{\bf E}\ifthenelse{\not\equal{}{#1}}{_{#1}}{}\tparen[#2\tparen]}

\DeclareMathAlphabet{\mathpzc}{OT1}{pzc}{m}{it}

\newcommand{\newagentvar}[3][\noaccents]{%
\expandafter\newcommand\expandafter{\csname #2\endcsname}{#1{#3}}%
\expandafter\newcommand\expandafter{\csname #2s\endcsname}{#1{\boldsymbol{#3}}}%
\expandafter\newcommand\expandafter{\csname #2smi\endcsname}[1][i]{#1{\boldsymbol{#3}}_{-##1}}%
\expandafter\newcommand\expandafter{\csname #2i\endcsname}[1][i]{#1{#3}\agind[##1]}%
\expandafter\newcommand\expandafter{\csname #2ith\endcsname}[1][i]{#1{#3}_{(##1)}}%
}

\newcommand{\newitemvar}[3][\noaccents]{%
\expandafter\newcommand\expandafter{\csname #2\endcsname}{#1{#3}}%
\expandafter\newcommand\expandafter{\csname #2s\endcsname}{#1{\boldsymbol{#3}}}%
\expandafter\newcommand\expandafter{\csname #2smj\endcsname}[1][j]{#1{\boldsymbol{#3}}_{-##1}}%
\expandafter\newcommand\expandafter{\csname #2j\endcsname}[1][j]{#1{#3}_{##1}}%
\expandafter\newcommand\expandafter{\csname #2jth\endcsname}[1][j]{#1{#3}_{(##1)}}%
}

\newagentvar{cost}{c}
\newagentvar[\constrained]{critcost}{\cost}
\newagentvar[\constrained]{price}{\cost}

\newcommand{\forbudget}[1]{#1_B}

\newcommand{\valB}{\forbudget{\val}}
\newcommand{\noisy}[1]{\tilde{#1}}

\newcommand{\mlext}{V}
\newcommand{\ccext}{V^+}

\newcommand{\opti}[1]{#1^\star}

\newcommand{\redval}{g}

\newcommand{\crs}{\pi}
\newcommand{\numSmall}{m}
\newagentvar{weight}{w}
\newagentvar{order}{\sigma}
\newagentvar{spent}{M}
\newagentvar{shap}{\psi}
\newagentvar{good}{G}
\newagentvar{minrank}{\Phi}
\newagentvar{dbt}{\mathcal{D}}
\newagentvar{costcurve}{C}
\newagentvar[\ironed]{costcurveh}{\costcurve}

\newagentvar{alloc}{x}

\newagentvar{quant}{q}
\newagentvar[\constrained]{exquant}{\quant}
\newagentvar[\constrained]{critquant}{\quant}  
\newagentvar[\optconstrained]{monoq}{\quant}
\newagentvar{Val}{V}
\newagentvar{toquant}{Q}

\newagentvar{qprice}{\price}
\newagentvar{qrev}{R}
\newagentvar[\ironed]{iqrev}{\qrev}

\newagentvar{qalloc}{y}
\newagentvar{cumalloc}{Y}
\newagentvar[\constrained]{calloc}{\qalloc}
\newagentvar[\constrained]{cumcalloc}{\cumalloc}
\newagentvar[\ironed]{ialloc}{\qalloc}
\newagentvar[\ironed]{icumalloc}{\cumalloc}

\newagentvar{rev}{R}
\newagentvar[\differentiated]{marg}{\rev}
\newagentvar{rawrev}{P}
\newagentvar[\differentiated]{rawmarg}{\rawrev}

\newagentvar[\primedarg]{inducedrev}{\rev}

\newagentvar[\tilde]{pseudorev}{\rev}
\newagentvar[\tilde]{pseudorawrev}{\rawrev}

\newagentvar{typespace}{T}
\newagentvar{typesubspace}{S}

\newagentvar{type}{t}
\newagentvar{othertype}{s}
\newagentvar{val}{v}
\newagentvar{outcome}{w}
\newagentvar{outcomespace}{{\cal W}}

\newagentvar{payment}{p}
\newagentvar{talloc}{\alloc}
\newagentvar[\tilde]{balloc}{\alloc}

\newagentvar{act}{a}
\newagentvar{bidspace}{A}
\newagentvar{actspace}{A}

\newagentvar[\constrained]{critval}{\val}
\newagentvar[\constrained]{crittype}{\type}
\newagentvar[\constrained]{critvirt}{\virt}
\newagentvar[\constrained]{reserve}{\val} 
\newagentvar[\constrained]{bidreserve}{\bid} 
\newagentvar[\optconstrained]{monop}{v}
\newagentvar[\optimized]{monorev}{\rev}

\newagentvar{rcalloc}{y}
\newagentvar[\optimized]{optrcalloc}{\rcalloc}
\newagentvar{biddist}{G}

\newagentvar[\constrained]{critbid}{\bid}
\newagentvar[\constrained]{cbid}{B}
\newagentvar[\primedarg]{wbid}{\bid}
\newagentvar{gfunc}{\vartheta}
\newagentvar{block}{C}

\newagentvar{util}{u}

\newagentvar{strat}{s}
\newagentvar{bid}{b}

\newagentvar{virt}{\phi}
\newagentvar{cumvirt}{\Phi}
\newagentvar{qvirt}{\phi}
\newagentvar[\ironed]{ivirt}{\virt}
\newagentvar[\ironed]{icumvirt}{\cumvirt}

\newagentvar{dist}{F}
\newagentvar{dens}{f}
\newagentvar{hazard}{h}
\newagentvar{cumhazard}{H}

\newagentvar[\ironed]{iprice}{\price}  
\newagentvar[\ironed]{ival}{\val}  

\newagentvar{ints}{{\cal I}}

\newagentvar{wal}{w}

\newitemvar{pos}{j}
\newitemvar[\differentiated]{mweight}{\weight}
\newitemvar{udtype}{\type}
\newitemvar[\constrained]{udcrittype}{\type}
\newitemvar{udalloc}{\alloc}
\newitemvar{udprice}{\price}
\newitemvar[\constrained]{udcalloc}{\qalloc}
\newitemvar[\constrained]{udcumcalloc}{\cumalloc}

\newagentvar{mech}{{\cal M}}
\newagentvar{alg}{{\cal A}}

\newcommand{\lagrange}{\lambda}

\newcommand{\budget}{B}

\newcommand{\reals}{{\mathbb R}}

\newagentvar{trans}{\sigma}

\newagentvar{demandset}{S}

\newcommand{\opt}{^\star}
\newcommand{\primed}{^\dagger}

\newagentvar{distout}{w}
\newagentvar{idistout}{\bar{w}}

\newcommand{\setdist}{{\cal D}}

\usepackage{fullpage}
\usepackage[ruled]{algorithm2e}

\SetAlFnt{\small}
\SetAlCapFnt{\small}
\SetAlCapNameFnt{\small}
\SetAlCapHSkip{0pt}
\IncMargin{-\parindent}

\begin{document}

\markboth{E. Balkanski and J. D. Hartline}{Bayesian Budget Feasibility with Posted Pricing}
\title{Bayesian Budget Feasibility with Posted Pricing}

\author{
  Eric Balkanski\\
  Harvard University \\
  School of Engineering and Applied Sciences\\
  \texttt{ericbalkanski@g.harvard.edu}
  \and
  Jason D. Hartline\\
  Northwestern University\\
  EECS Department \\
  \texttt{hartline@eecs.northwestern.edu}
}

\newtheorem{theorem}{Theorem}%

\newaliascnt{lemma}{theorem}
\newtheorem{lemma}[lemma]{Lemma}%
\aliascntresetthe{lemma}
\providecommand*{\lemmaautorefname}{Lemma}


\newaliascnt{corollary}{theorem}
\newtheorem{corollary}[corollary]{Corollary}%
\aliascntresetthe{corollary}
\providecommand*{\corollaryautorefname}{Corollary}

\newaliascnt{proposition}{theorem}
\newtheorem{proposition}[proposition]{Proposition}%
\aliascntresetthe{proposition}
\providecommand*{\propositionautorefname}{Proposition}


\newtheorem{definition}{Definition}
\newcommand{\definitionautorefname}{Definition}

\newtheorem{example}{Example}
\newcommand{\exampleautorefname}{Example}

\maketitle

\begin{abstract}
We consider the problem of budget feasible mechanism design proposed by \citet{sin-10}, but in a Bayesian setting.  A principal has a public value for hiring a subset of the agents and a budget, while the agents have private costs for being hired.  We consider both additive and submodular value functions of the principal.  We show that there are simple, practical, ex post budget balanced posted pricing mechanisms that approximate the value obtained by the Bayesian optimal mechanism that is budget balanced only in expectation.  A main motivating application for this work is crowdsourcing, e.g., on Mechanical Turk, where workers are drawn from a large population and posted pricing is standard. Our analysis methods relate to contention resolution schemes in submodular optimization of \citet{VCZ-11} and the correlation gap analysis of \citet{yan-11}.
\end{abstract}

\section{Introduction}

Consider the problem of hiring workers to
complete complex tasks on crowdsourcing platforms such as Mechanical
Turk.  A principal must select a set of participants, henceforth
agents, whose contributions will be aggregated to complete the task.
The principal's value for the task is a function of the set of
participants selected and the principal's budget limits the total
payments to participants.  We assume that the principal's value is
submodular, i.e., it exhibits diminishing returns to recruiting
additional participants.  The participants have a private cost for
participating and will choose to participate strategically to optimize
their payments received relative to this cost.  The principal seeks a
budget feasible mechanism for selecting participants so as to maximize
the value of the completed task.


The literature on {\em budget feasible mechanism design} initiated by
\citet{sin-10} studies this problem; however, it primarily considers
sealed-bid mechanisms which do not tend to be seen on crowdsourcing
platforms like Mechanical Turk.  Instead, these platforms use posted
pricing mechanisms.  We follow a traditional economics approach to
this problem where agents' costs are drawn from a common prior
distribution and a mechanism is sought to optimize the principal's
value function in expectation.  Note that this approach is especially
relevant to the principal's problem as the workers on crowdsourcing
platforms are drawn from a large population of available workers.
We show that posted pricing mechanisms give a good
approximation to the optimal sealed-bid mechanism.  Additionally, we
give efficient algorithms for calculating the appropriate prices. In comparison to other work in optimization of prices in
crowdsourcing, our work focuses on the use of prices to control
participation and not the level of effort of participants.  Controlling
the level of effort of participants was studied in online behavioral
experiments by \citet{HSSV-15}, theoretically for crowdsourcing
contests by \citet{CHS-12}, and for user generated content by
\citet{ISS-15}.

\paragraph{Overview of Approach.}

Our approach follows similarly to that of \citet{ala-14} and
\citet{yan-11}.  The starting point for our analysis is an upper bound
on the performance of the optimal sealed bid mechanism that relaxes
the ex post budget constraint on the mechanism to hold ex ante, i.e., in
expectation over the private costs of the agents.  Via this {\em ex ante
  relaxation} and the \citet{mye-81} theory of virtual values, we
construct a posted price mechanism that is budget feasible in
expectation and a $1 - 1/e$ approximation to the optimal ex ante
mechanism when the principal's value function exhibits decreasing
returns, i.e., is {\em submodular}.  For the special case
where the principal's value function is {\em additive}, this posted pricing
is optimal (for the ex ante relaxation).

We then consider posting the prices from the solution to the ex ante
relaxation until the budget runs out.  The resulting mechanism is ex
post budget feasible, but suffers a loss in performance because the
budget may run out early.  The main technical contribution of this
work is to show that the performance of such a price
posting mechanism compares favorably to the optimal sealed-bid
mechanism. Previous work in mechanism design gives techniques which are now well understood to satisfy ex post allocation constraints. Ex post payment constraints require different techniques and our analyses follow two basic approaches that combine optimization and mechanism design concepts.  To analyze the
performance of the posted pricing under any arrival order of the
agents, we solve the ex ante relaxation with a slightly smaller budget
and then, using results from the \citet{VCZ-11} analysis of {\em
  contention resolution schemes}, show that it is unlikely for the
original ex post budget constraint to bind. Alternatively, we obtain better bounds
for additive value functions and when the order of agent arrivals can
be specified by the mechanism via the {\em correlation gap} approach
of \citet{yan-11}.  As a corollary, we obtain new correlation gap results for integral and fractional knapsack set functions.  Moreover, when the environment is symmetric (both in
distribution of agent costs and the principal's value function), the
submodular case can be reduced to the additive case.

The prices identified above can be computed or approximately computed
in polynomial time.  In particular, for submodular value functions, we
reduce the problem of finding the prices to the well-known {\em greedy
  algorithm} for submodular optimization.  The identified prices
approximate the optimal prices with relative loss in the value
function that is within a factor of $1-1/e$.  For additive value
functions, the optimization problem simplifies to a monopoly pricing
problem of classic microeconomics. Similarly to the \citet{MS-83}
treatment of welfare maximization subject to budget balance in a
buyer--seller exchange, optimization in this context is based on {\em
  Lagrangian virtual surplus}.  These optimal prices can be
approximated arbitrarily precisely by solving this problem on a
discretized instance.





\paragraph{Related work.}

The prior literature on budget feasibility primarily considers a
worst-case design and analysis framework that compares the performance
of the designed mechanism to the {\em first-best} outcome, i.e., the
one that could be obtained if the agents' costs were public.  See
\citet{sin-10}, \citet{BCGL-12}, \citet{BKS-12}, and \citet{AGN-14}.
Our analysis compares the designed mechanism, in expectation for the
known prior distribution, to the {\em second-best} outcome, i.e., the
one obtained by the Bayesian optimal mechanism.  

The following results are for prior-free mechanisms in comparison to
the first-best outcome.  \citet{sin-10} obtained a randomized truthful
budget feasible mechanism with a constant factor approximation for
submodular value functions, \citet{CGL-11} then improved the analysis
of this mechanism to a 0.13 approximation. In the Bayesian setting,
\citet{BCGL-12} obtained a constant approximation for subadditive
functions. More recently, \citet{AGN-14} obtained better bounds by
considering large markets, which we also consider in this
paper. Finally, \citet{BKS-12} also considered posted pricing
mechanisms but when the agents arrive online.  They obtained a constant
approximation for the class of symmetric submodular functions. They
also obtained a $O(\log n)$ mechanism for the case of submodular
functions.  In comparison to this last paper, we give much better
bounds when the prior distribution on costs is known.

The starting point for our analysis is the solution to the relaxed
problem of budget balance in expectation, i.e., ex ante.  In the
additive case, this problem was recently studied by \citet{EG-14}.
They show that posted pricing mechanisms solve the relaxed problem and
remark that the same performance can be obtained with ex post budget
balance, but at the expense of relaxing ex post individual rationality
(for the bidders) and not with a posted pricing.  This latter
observation follows, for example, by applying a general construction
of \citet{EF-99}.  Our analysis of the relaxed problem gives a much simpler proof of
their main theorem.

Budget feasibility has also been studied in the context of
crowdsourcing. Among that line of work, the model considered in
\citet{AGN-14} is the closest to ours, and will be compared in detail
below.  \citet{SK-13} and \citet{SM-13} consider the special case of
our model where the principal's value function is the number of tasks
performed.  The former studies posted pricing for agents with i.i.d.\@
costs from an unknown distribution, while the latter studies sealed
bid mechanisms without a prior.

\paragraph{Our results.}

\begin{figure}
\begin{center}
\scalebox{0.75}{
\renewcommand{\arraystretch}{1.4}
\begin{tabular}{ |c|c|c|c|c |}
\hline
\rule{0pt}{3ex} Value Function & \begin{tabular}{c}Mechanism \\ Family\end{tabular} & \begin{tabular}{c}Ex Post Constraint \\ Approach\end{tabular} & General Result & \begin{tabular}{c}Large \\  Markets \end{tabular}  \rule[-1.2ex]{0pt}{0pt} \\
\hline \hline

Additive & \rule[1.2ex]{0pt}{0pt}  \begin{tabular}{c}Sequential\\ Posted Pricing\end{tabular}
                   & \begin{tabular}{c}Correlation \\  Gap \end{tabular} & $(1-\frac{1}{\sqrt{2 \pi k}})(1 - \frac{1}{k}) $& $1$\\\hline

\begin{tabular}{c}Symmetric \\ Submodular \end{tabular} & \rule[1.2ex]{0pt}{0pt}  \begin{tabular}{c}Oblivious\\ Posted Pricing\end{tabular}
                   & \begin{tabular}{c}Correlation \\  Gap \end{tabular} & $(1-\frac{1}{\sqrt{2 \pi k}})(1 - \frac{1}{k}) $& $1$\\\hline

 \begin{tabular}{c}Submodular \\ (computational) \end{tabular}  &\begin{tabular}{c}Oblivious \\ Posted Pricing\end{tabular} & \begin{tabular}{c}Contention \\  Resolution \end{tabular}
 & $(1 - \frac{1}{e })^2 ( 1 - \epsilon) (1 - e^{- \epsilon^2 (1 - \epsilon) k / 12})$ &  $(1 - \frac{1}{e })^2 \approx 0.40$ \\\hline
 
 \begin{tabular}{c}Submodular \\ (non-computational) \end{tabular} & \begin{tabular}{c}Oblivious \\ Posted Pricing\end{tabular} & \begin{tabular}{c}Contention \\  Resolution \end{tabular}
 & $(1 - \frac{1}{e }) ( 1 - \epsilon) (1 - e^{- \epsilon^2 (1 - \epsilon) k / 12})$ & $1 - \frac{1}{e } \approx 0.63$ \\\hline
\end{tabular}}
\caption{Our results are approximations to the
  Bayesian optimal mechanism.  Bounds are parameterized by the market size $k$, a lower bound on the number of
  agents that can be simultaneously selected with the given budget (see
  Definition \ref{def:largemarkets}). In large markets, $k$ grows large.  The given results with the contention resolution approach require $k \geq 4$ and
  $\epsilon \in (2 /k, 1/2)$, a result for $k < 4$ is mentioned in Section~\ref{s:crs}.. For the symmetric submodular results, we also assume symmetric distributions on costs. Our
computational results also have an additional
$o(1)$ loss due to discretization.} 
\end{center}
\label{f:results-table}
\end{figure}

Our results are summarized in Figure~\ref{f:results-table}. We
consider two main classes of valuation functions, additive and
submodular. We use two different methods to satisfy the ex post payment constraint, one is based on contention resolution schemes and the
other on correlation gap. Contention resolution schemes give an
oblivious posted price mechanism, i.e., one that obtains its proven
bound under any arrival order of the agents.  The correlation gap
approach, for the case where the principal has an additive value
function, gives a sequential posted price mechanism. Such a mechanism
is specified by an ordering on agents and take-it-or-leave-it prices
to offer each agent.  As a special case, we consider symmetric
environments where both the value function and the distribution is
symmetric.

Our results can most directly be compared to those of
\citet{AGN-14}, but with the following caveats. Their results are for
sealed bid mechanisms while ours are for posted pricings; their
mechanism is prior-free while ours is parameterized by the prior
distribution on agent costs; their results compare performance to the
first-best outcome, i.e., without incentive constraints, while ours
compare to the second-best outcome, i.e., that of the Bayesian optimal
mechanism (with incentive constraints). They obtain approximation ratios of $1 - 1/e$, $1/3$ and $1/2$ in large markets respectively for additive, submodular (computational), and submodular (non-computational) value functions. Moreover, they show that no truthful mechanism can achieve an approximation ratio better than $1 - 1/e$ with respect to the first-best outcome for additive value functions.

\paragraph{Discussion about posted pricing mechanisms and benchmarks.}
Following a line of literature in mechanism design that was initiated
by \citet{CHMS-10}, the goal of this work is to show that there exists simple posted pricing mechanisms that approximate the optimal sealed-bid mechanism. Two quantities of interest therefore need to be separated. The first is the cost of incentive compatibility in budget feasible settings, i.e., the gap between the first-best and second-best benchmarks. The second is the cost of simplicity, i.e., the loss of a posted pricing mechanism compared to the Bayesian optimal mechanism. Prior work with comparisons to a first-best benchmark has approximations that are a combination of both of these quantities. Our comparison to the
second-best outcome isolates the loss from a simple decentralized
pricing over the optimal centralized mechanism as the quantity of
interest.

\paragraph{Paper Organization.}

We start with preliminaries in Section~\ref{s:prelim} to introduce the
model and different concepts used in this paper. We then
describe posted price mechanisms for the ex ante relaxation, where the
budget holds in expectation, in Section~\ref{s:relaxations}. We
explain how to go from an ex ante posted price mechanism to an ex post
posted price mechanism using two different methods, one inspired by
contention resolution schemes in Section~\ref{s:crs} and another based
on a correlation gap analysis in Section~\ref{s:additive}. We
tackle the computation issues of finding a good ex ante mechanism in
Section~\ref{s:optimal}. In Section~\ref{s:sym}, we study symmetric environments. Up to Section~\ref{s:sym}, cost distributions are assumed to be regular and Section~\ref{s:irregular} considers the case where some distributions might be irregular. Throughout the paper, we assume that the principal's valuation
function is monotone and submodular.

\section{Preliminaries}
\label{s:prelim}

There are $n$ agents $N = \{1,\ldots,n\}$.  Agent $i$ has a private
cost $\costi$ for providing a service that is drawn from a
distribution $\disti$ (denoting the cumulative distribution function)
with density $\densi$.  Indicator variable $\alloci$ denotes whether
or not $i$ provides service and $\paymenti$ denotes the payment $i$
receives.  Agent $i$ aims to optimize her utility given by $\paymenti
- \costi \alloci$.  The cost profile is denoted $\costs =
(\costi[1],\ldots,\costi[n])$; the joint distribution on costs is the
product distribution $\dists = \disti[1] \times \cdots \times
\disti[n]$; the payment profile is denoted $\payments =
(\paymenti[1],\ldots,\paymenti[n])$; and the allocation profile is
denoted $\allocs = (\alloci[1],\ldots,\alloci[n])$.

The principal has a value function $\val : \{0,1\}^n \to \reals_{+}$. For
allocation profile $\allocs \in \{0,1\}^n$ or set of agents $S = \{i :
\alloci = 1\}$ who provide service, the value to the principal is
$\val(\allocs) = \val(S)$.  The principal has a budget $\budget$ and
requires the payments to the agents not to exceed the budget, i.e.,
$\sum_i \paymenti \leq \budget$. The following mathematical program
captures the principal's objective.
\begin{align}
\label{eq:main}
\max_{\allocs,\payments}\ &\expect[\costs]{\val(\allocs(\costs))}\\
\notag    \text{s.t. } &\sum\nolimits_i \paymenti(\costs) \leq \budget \hspace{0.5 cm} \forall \costs,\\
\notag                 &\text{$\allocs(\cdot)$ and $\payments(\cdot)$ are incentive compatible.}
\end{align}

We consider only mechanisms that are incentive compatible.  A
mechanism is {\em incentive compatible} (IC), if truthful reporting of
the agents is a dominant strategy equilibrium.\footnote{The
  restriction to dominant strategy mechanisms over Bayesian incentive
  compatible mechanisms is without loss for the budget feasibility
  objective.}  We will consider the budget constraint both ex ante,
i.e., in expectation over realizations of agents' costs and random
choices of the mechanism, and ex post, i.e., the payments to the
agents never exceed the budget.  The main goal of this paper is to
approximate the optimal ex ante budget feasible mechanism with an ex
post budget feasible posted pricing mechanism.  Posted pricing
mechanisms are trivially incentive compatible.

\begin{definition}
\label{d:posted-pricing}
The {\em posted pricing} $(\critcosts,\orders)$, for prices $\critcosts$ and ordering on agents $\orders$, is:
\begin{enumerate}
\item The remaining budget is initially $\budget$.
\item The agents arrive in order $\orders$.
\item If agent $i$ arrives with cost $\costi$ below her offered price
  $\critcosti$ which is below the remaining budget, then select this agent for service, pay her $\critcosti$, and deduct $\critcosti$ from the remaining budget. Otherwise, discard this agent. 
\end{enumerate}
\end{definition}

For (implicit) distribution on costs $\dists$, we can equivalently
specify a posted pricing $(\critcosts,\orders)$ as
$(\exquants,\orders)$ where $\exquanti = \disti(\critcosti)$ is the
{\em marginal probability} that agent $i$ with cost $\costi \sim
\disti$ would accept the price $\critcosti$.\footnote{It is common in
  Bayesian mechanism design to consider the agents' private costs in
  {\em quantile space} where $i$'s quantile $\quanti = \disti(\costi)$
  is the measure of cost lower than $\costi$ according to $\disti$.
  Agent quantiles are always uniformly distributed on $[0,1]$.  From
  this perspective, $\exquanti$ is agent $i$'s price in quantile
  space.} 

Note that the prices $\critcosts$ are non-adaptive, i.e., fixed before
the agents arrive.  We consider posted pricing mechanisms under two
different models for agent arrival.  In the {\em sequential posted
  pricing} model, the ordering $\orders$ can be fixed in advance by
the mechanism and, without computational considerations, our analysis
is for the best case ordering of the prices.  In the {\em oblivious
  posted pricing} model, the ordering $\orders$ is unconstrained and
our analysis is worst case with respect to this ordering.  An
oblivious posted pricing is denoted $\critcosts$.  We
compare our mechanisms to an {\em ex ante posted pricing}
$\critcosts$ where the budget constraint holds in expectation, i.e.,
$\sum_i \critcosti\,\critquanti \leq \budget$.  The value of an ex
ante posted pricing is $\expect[S\sim\critquants]{\val(S)}$ where
$S\sim\critquants$ adds each agent $i$ to $S$ independently with
probability $\critquanti$.




The paper focuses on value functions that are monotone and submodular
(Definition~\ref{d:submodular}).  An important special case, which we
will treat separately, is that of {\em additive value functions} where
each agent has a value $\vali$ and the value function is $\val(S) =
\sum_{i\in S} \vali$.
\begin{definition}
\label{d:submodular}
A set function $\val : \{0,1\}^n \to \reals_+$ is {\em monotone submodular} if
\begin{itemize}
\item (monotonicity) $\val(T) \leq \val(S)$ for all $T \subset S$, and
\item (submodularity) for all $T \subset S$ the marginal contribution of $i \not \in S$ to $T$ is at least its marginal contribution to $S$.  In other words,
$$
\val(T \cup \{i\}) - \val(T) \geq \val(S \cup \{i\}) - \val(S).
$$
\end{itemize} 
\end{definition}

Our analysis is based on the relationship between a set function and two
standard extensions of a set functions from the domain $\{0,1\}^n$ to
the domain $[0,1]^n$.  For submodular set functions, these extensions
were studied by \citet{CCPV-07} and \citet{ADSY-10}.
\begin{definition}
\label{d:submodular-extensions}
Given a set function $\val : \{0,1\}^n \to \reals_+$,
\begin{itemize}
\item its {\em concave closure} $\ccext(\cdot)$ (a.k.a., {\em
  correlated value}) is the smallest concave function that upper
  bounds the set function.  Alternatively, $\ccext(\exquants) =
  \max_{\setdist} \expect[S \sim \setdist]{\val(S)}$ with the
  maximization taken over all distributions $\setdist$ with marginal probabilities
  $\exquants = (\exquanti[1],\ldots,\exquanti[n])$; and
\item its {\em multilinear extension} $\mlext(\cdot)$ (a.k.a., {\em
  independent value}) is the expected value of the set function when
  each element $i$ is drawn independently with marginal probability
  $\exquanti$.  In other words, $\mlext(\exquants) = \expect[S \sim
    \exquants]{\val(S)}$.
\end{itemize}
\end{definition}

For any set function, the concave closure is clearly an upper bound on
the multilinear extension.  For submodular functions the inequality
approximately holds in the opposite direction as well.  By the
interpretation of the multilinear extension as the expected value of
the set function for independent distribution and the concave closure
as the expected value of the set function for correlated
distributions, their worst case ratio over marginal probabilities
$\exquants$ is known as the {\em correlation gap} \citep{ADSY-10}.


\begin{theorem}[\citealp{CCPV-07}, \citealp{ADSY-10}] 
\label{t:cg-submodular}
For monotone submodular set function $\val(\cdot)$, the correlation gap is 
$$\min_{\exquants} \frac{\mlext(\exquants)}{\ccext(\exquants)} \leq 1 - 1/e.$$
\end{theorem}

\begin{theorem}[\citealp{yan-11}]
\label{t:cg-multi-unit} 
For a $k$-highest-value-elements set function $\val(\cdot)$, which is
additive with value $\vali$ for element $i$ up to a capacity of at most $k$
elements, the correlation gap is $$\min_{\exquants} \frac{ \mlext(\exquants) }{
\ccext(\exquants)} \leq 1-1/\sqrt{2 \crs k}.$$
\end{theorem}




Our analysis is parameterized by a measure of the size of the
market.  This notion of market size is standard in the literature,
e.g., see \citet{BCGL-12} and \citet{AGN-14}.  A large market analysis
considers the market size in the limit. Although large markets are described as an assumption by \citet{AGN-14}, the market size $k$ is a parameter in our analysis and we obtain results for any market size.

\begin{definition} \label{def:largemarkets}
A market is {\em $k$-large} for prices $\critcosts$ and budget $B$ if
$B / \critcosti \geq k$ for all agents $i$.
\end{definition}
Note that the market size depends on prices and therefore on the mechanism, which is inherent to our analysis. These prices can trivially be upper bounded by the maximum cost that can be drawn from the distributions. 

\section{The Ex Ante Budget Feasible and Concave Closure Relaxations}
\label{s:relaxations}

In this section we relax the objective function and the budget
constraint to make the problem more amenable to optimization. We first
relax the budget constraint so that it only holds in expectation,
making it an ex ante feasibility constraint. We then upper bound the
value function by its concave closure.  With an ex ante feasibility
constraint, the objective is to optimize the following ex ante program over
allocation rule $\allocs(\cdot)$ and payment rule $\payments(\cdot)$
with $\costs \sim \dists$.
\begin{align}
\label{eq:bayes-opt-prog}
\max_{\allocs,\payments}\ &\expect[\costs]{\val(\allocs(\costs))}\\
\notag    \text{s.t. } &\sum\nolimits_i \expect[\costs]{\paymenti(\costs)} \leq \budget,\\
\notag                 &\text{$\allocs(\cdot)$ and $\payments(\cdot)$ are IC.}
\end{align}

When payments are part of the principal's objective or constraints,
the Bayesian mechanism design problem will typically rely on the
\citet{mye-81} theory of virtual values or, in our case where the
agents are sellers, virtual costs.  The {\em virtual cost} of agent
$i$ with cost $\costi$ drawn from distribution $\disti$ is
$\virti(\costi) = \costi + \frac{\disti(\costi)}{\densi(\costi)}$.
The {\em virtual surplus} of an agent $i$ with virtual cost
$\virti(\costi)$ and allocation indicator $\alloci$ is
$\virti(\costi)\,\alloci$.

\begin{lemma}[\citealp{MS-83}] 
\label{l:payment=virtual-cost}
In any incentive compatible mechanism, any agent $i$'s expected
payment is equal to her expected virtual surplus, i.e., for $\costs
\sim \dists$,
$$\expect[\costs]{\paymenti(\costs)} = \expect[\costs]{\virti(\costs)\,\alloci(\costs)}.$$
\end{lemma}

The definition of virtual costs and Lemma~\ref{l:payment=virtual-cost}
allows the ex ante program~\eqref{eq:bayes-opt-prog} to be rewritten in terms
of the allocation rule only.  To do so, we invoke the following
characterization of incentive compatible mechanisms of
\citet{mye-81}.

\begin{lemma}[\citealp{mye-81}]
\label{l:IC=monotone}
There exists an incentive compatible mechanism with allocation rule
$\allocs(\cdot)$ if and only if $\allocs(\cdot)$ is monotone in the
cost of any agent.
\end{lemma}

We now rewrite the optimization program~\eqref{eq:bayes-opt-prog} by
substituting in virtual costs for
payments to obtain the following virtual surplus program,

\begin{align}
\label{eq:relaxed-bayes-opt-prog}
\max_{\allocs}\ &\expect[\costs]{\val(\allocs(\costs))}\\
\notag    \text{s.t. } &\sum\nolimits_i \expect[\costs]{\virti(\costs)\,\alloci(\costs)} \leq \budget, \\
\notag & \allocs(\cdot) \text{ is monotone in the cost of any agent.}
\end{align}

 For the general case of submodular value functions, the expected
 value of the set function $\val(\cdot)$ is upper bounded by its
 concave closure (Definition~\ref{d:submodular-extensions}) as
 follows.  The allocation rule $\allocs(\cdot)$ that optimizes this
 virtual surplus program induces, for $\costs \sim \dists$, a
 distribution over sets of winning agents.  Denote this distribution
 by $\setdist$ and denote by $\exquants$ the profile of marginal
 probabilities, i.e., with $\exquanti = \prob[S \sim \setdist]{ i \in
   S}$.  By the definition of the concave closure of the set function
 $\val(\cdot)$, $\expect[\costs]{\val(\allocs(\costs))} =
 \expect[S\sim \setdist]{\val(S)} \leq \ccext(\exquants)$.

The payment to an agent is lower bounded by the payment from price
posting.  As above, the optimal mechanism selects agent $i$ with
probability $\exquanti$.  When virtual costs are monotonically
increasing, i.e., in the case of \emph{regular distributions}, the expected
payment to an agent $i$ selected with probability $\exquanti$ is
minimized if agent $i$ is served if and only if $\costi \leq
\disti^{-1}(\exquanti)$ by Lemma~\ref{l:payment=virtual-cost} since
these costs minimize $\virti(\costs)$.\footnote{The case of irregular distributions is considered in Section~\ref{s:irregular}.}  Thus, the mechanism that minimizes expected
payments and serves each agent $i$ with probability $\exquanti$ is the
mechanism that posts price $\pricei = \disti^{-1}(\exquanti)$ to each
agent $i$.

\begin{lemma} 
For any agent with cost drawn from regular distribution $\disti$ and
any incentive compatible mechanism that selects agent $i$ with
probability $\exquanti$, the expected payment of agent $i$ is at least
$\exquanti \critcosti$ where $\critcosti = \disti^{-1}(\exquanti)$.
\end{lemma}

Combining the relaxation of the value function and the relaxation of
the payments we obtain the following concave closure program,
\begin{align}
\label{eq:relaxed-concave-bayes-opt-prog}
\max_{\quants}\ &\ccext(\quants) \\
\notag    \text{s.t. } &\sum\nolimits_i \quanti \disti^{-1}(\quanti) \leq \budget.
\end{align}

\begin{lemma}\label{l:ub} Let $\exquants^+$ be the optimal solution to the concave closure program~\eqref{eq:relaxed-concave-bayes-opt-prog}, then $\ccext(\exquants^+)$ upper bounds the performance of the optimal ex ante mechanism in the case of regular cost distributions.
\end{lemma}

Posted price mechanisms are trivially incentive compatible.  Since the distributions of agents' costs are independent, the set of agents who will accept their offer with a posted price mechanism is a set which will contain each agent with some probability $\quanti$ independently. Therefore the performance of a posted price mechanism where agents accept their offer with probabilities $\quants$ is the multilinear extension $\mlext(\quants)$. This motivates us to rewrite the concave closure program~\eqref{eq:relaxed-concave-bayes-opt-prog} as the following multilinear extension program,
\begin{align}
\label{eq:multilinear-prog}
\max_{\quants}\ &\mlext(\quants) \\
\notag    \text{s.t. } &\sum\nolimits_i \quanti \disti^{-1}(\quanti) \leq \budget.
\end{align}
Maximizing the multilinear extension program gives us an ex ante posted price mechanism that is approximately optimal.

\begin{theorem}
\label{thm:mlextccext}  In the case of monotone submodular value functions and regular cost distributions, the ex ante mechanism that posts price $\critcost_i = \disti^{-1}(\exquanti)$ to each agent $i$ is an $1  - 1/e$ approximation to the optimal ex ante mechanism, where $\exquants$ is the optimal solution to the multilinear extension program~\eqref{eq:multilinear-prog}.
\end{theorem} 
\begin{proof}
Let $\exquants^{+}$ be the optimal solution to the concave closure
program~\eqref{eq:relaxed-concave-bayes-opt-prog}. By
Theorem~\ref{t:cg-submodular}, $\mlext(\exquants^{+}) \geq (1 -
1/e)\ccext(\exquants^{+})$. By the optimality of $\exquants$,
$\mlext(\exquants) \geq \mlext(\exquants^{+})$. Since the performance
of posting price $\disti^{-1}(\exquanti)$ to each agent $i$ is
$\mlext(\exquants)$ and since $\ \ccext(\exquants^{+})$ upper bounds
the performance of the optimal ex ante mechanism by Lemma~\ref{l:ub}, posting price
$\disti^{-1}(\exquanti)$ to each agent is an $1 - 1/e$ approximation
to the optimal ex ante mechanism.
\end{proof}

Note that in the additive case where each agent has value $\vali$,
$\mlext(\quants) = \ccext(\quants) = \sum_i \vali \quanti$ and we get
the following corollary.
\begin{corollary} \label{c:optexante} In the case of additive value functions and regular cost distributions, the ex ante mechanism that posts price $\critcost_i = \disti^{-1}(\exquanti)$ to each agent $i$ is an optimal mechanism, where $\exquants$ is the optimal solution to the multilinear extension program~\eqref{eq:multilinear-prog}.
\end{corollary}

We discuss the computational issues of finding a good solution
$\quants$ to the multilinear extension
program~\eqref{eq:multilinear-prog} in Section~\ref{s:optimal}. For
the case of submodular functions, we reduce the problem to submodular
function maximization (with a cardinality constraint) for which the greedy
algorithm gives an $1-1/e$ approximation.  In the additive case, we
will show that the optimal ex ante budget feasible mechanism can be
found by taking the Lagrangian relaxation of the virtual surplus
program~\eqref{eq:relaxed-bayes-opt-prog}.

\section{Submodular Value and Oblivious Posted Pricing}
\label{s:crs}

In the previous section, we obtained an ex ante mechanism by
optimizing the multilinear extension
program~\eqref{eq:multilinear-prog}. In this section we analyze the
performance of oblivious posted pricing (with an ex post budget
constraint).

The approach of this section is the following: lower the budget by
some small amount and optimize the multilinear extension
program~\eqref{eq:multilinear-prog} so that the lowered budget is
satisfied ex ante.  With the budget sufficiently lowered, with high
probability the cost (sum of prices) of the set of agents who would
accept their offer is under the original budget (regardless of their
arrival order and ex post).

This approach is a special case of that taken by the contention
resolution schemes of \citet{VCZ-11} and we first review some known
bounds.  The first comes from the submodularity of the value function; the second comes from the Chernoff bound.

\begin{theorem}[\citealp{BKNS-10}] 
\label{t:cr-main}
Given a non-negative monotone submodular function $\val(\cdot)$, a random set $R$ which contains each agent $i$ independently with probability $\exquanti$, and a (possibly randomized) procedure $\pi$
that maps (possibly infeasible) sets to feasible sets such that,
\begin{itemize}
\item (marginal property) for all $i$, $\prob[R\sim \exquants;\pi]{i \in \pi(R) \given i
  \in R} \geq \gamma$, and
\item (monotonicity  property) for all $T \subseteq S$ and $i \in T$, $\prob[\pi]{i \in \pi(T)} \geq \prob[\pi]{i \in \pi(S)}$,
\end{itemize}
then $\expect[R\sim \exquants;\pi]{\val(\pi(R))} \geq \gamma \cdot \expect[R\sim \exquants]{\val(R)}$.
\end{theorem}

\begin{theorem}[\citealp{VCZ-11} \footnote{The formulation of this theorem is slightly different than in \cite{VCZ-11} but follows easily from their analysis.}]
\label{t:CR-knapsack}
Given $\epsilon \in (0,1/2)$, budget $\budget$, independent variables  $p_i$ that are the payments to each agent 
such that,
\begin{itemize}
\item (scaled ex ante budget constraint) $\sum_i \expect{p_i}
  \leq (1-\epsilon)\, \budget$, 
  \item ($k$-large market) $p_i$ is bounded by $[0, \budget / k]$ for all $i$, and
\item $ k > 2 / \epsilon$,
\end{itemize}
then the probability that the sum of costs of selected agents does not exceed
the budget less the cost of any agent, i.e., $\prob{\sum_{i } p_i\leq (1-1/k) \budget}$, is at least
$1 - e^{-\epsilon^2(1-\epsilon)k/12}$.
\end{theorem}

We now connect these two results by relating the probability that the
sum of costs does not exceed $(1-1/k) \budget$ of Theorem~\ref{t:CR-knapsack}
to $\gamma$ of Theorem~\ref{t:cr-main} and then show that posted
pricings satisfy the conditions of Theorem~\ref{t:cr-main}.

\begin{lemma}
\label{l:budget-offer-relation}
For sequential posted pricing $(\prices,\orders)$ that satisfy the scaled ex ante budget constraint and $k$-large market conditions, the probability
that an agent is offered her price is lower bounded by $\prob[R \sim
  \exquants]{\sum_{i \in R} \critcosti \leq (1-1/k) \budget}$, the probability that the sum of
the prices of agents who would accept their offered price is at most
$(1-1/k)\budget$.
\end{lemma}

\begin{proof}
If the total cost of all agents who would accept their price is at
most $(1-1/k)\budget$ then this budget remains at the time an agent
$i$ is considered in the sequence $\orders$.  By the definition of $k \geq
\budget / \pricei$ it is feasible to serve this agent and so she is
offered her price $\pricei$ by the sequential posted pricing mechanism.
\end{proof}

\begin{lemma}\label{l:CR-main}
For sequential posted pricing $(\exquants,\orders)$, if each agent is
offered her price with probability at least $\gamma$, then the
expected value of the mechanism is at least $\gamma
\mlext(\critquants)$.
\end{lemma}

\begin{proof}
It suffices to show, for sequential posted pricing
$(\exquants,\orders)$ with an ex post budget constraint $\budget$,
that the marginal and monotonicity properties of
Theorem~\ref{t:cr-main} hold. 

In our case, $R\sim \exquants$ is the random set of agents who would accept
their offer if the budget never runs out.  Given a set of agents $R$
who accept their offer, define $\pi(R)$ to be the set of agents who
accept their offer and who arrive before the budget runs out. In our
case, $\pi$ is deterministic given the ordering $\orders$. Note that
$\prob[R\sim \exquants;\pi]{i \in \pi(R) \given i \in R}$ is equal to
the probability that an agent gets offered her price, meaning that she
arrives before the budget runs out. Thus, by the assumption of the
lemma the marginal property holds.

For the monotonicity property, consider two sets $T \subseteq S$. When
an agent $i$ arrives in the posted price mechanism, the mechanism has
spent less if the set of agents who accept their offer is $T$ than if
this set is $S$. Therefore $i \in \pi(S)$ implies that $i \in \pi(T)$ and the monotonicity property holds.
\end{proof}

By combining the previous results, we obtain the main theorem for this section.

\begin{theorem} 
For $\epsilon \in (0,1/2)$, if the oblivious posted pricing $\prices$
corresponding to the optimal solution $\exquants$ to the multilinear
extension program~\eqref{eq:multilinear-prog} with budget $(1 -
\epsilon) \budget$
(i.e., with $\price_i = \disti^{-1}(\exquanti)$ for each agent $i$) 
satisfies $2/ \epsilon \leq k \leq B / \max_i \pricei$, then this
posted pricing mechanism is a $(1 - 1 /e)(1 - \epsilon) (1- e^{-
  \epsilon^2 ( 1- \epsilon) k / 12})$ approximation to the optimal
mechanism for submodular value functions and $(1 - \epsilon) (1- e^{-
  \epsilon^2 ( 1- \epsilon) k / 12})$ for additive value functions in the case of regular cost distributions.
\end{theorem} 
\begin{proof}
The proof starts with the ex ante mechanism from the previous section and then applies results from this section to modify it into an ex post mechanism.

Let $\exquants$ be the optimal solution to the multilinear extension
program~\eqref{eq:multilinear-prog} with budget $(1 - \epsilon)
\budget$, $\exquants^+_{(1-\epsilon) \budget}$ be the optimal solution to the concave closure program~\eqref{eq:relaxed-concave-bayes-opt-prog} with budget $(1 - \epsilon)
\budget$, and $\exquants^+_{\budget}$ be the optimal solution to the concave closure
program~\eqref{eq:relaxed-concave-bayes-opt-prog} with budget $
\budget$.

 By the optimality of $\exquants$ and Theorem~\ref{t:cg-submodular},
$$ 
\mlext(\exquants) \geq \mlext(\exquants^+_{(1-\epsilon) \budget}) \geq (1 - \tfrac{1}{e}) \ccext(\exquants^+_{(1-\epsilon) \budget}).$$ 
Note that the solution $(1 - \epsilon)\exquants^+_{\budget}$ has cost at most $(1 - \epsilon) \budget$ since $\disti^{-1}(\cdot)$ is increasing. So by the optimality of $\exquants^+_{(1-\epsilon) \budget}$ and by the concavity of the concave closure $\ccext(\cdot)$, $$\ccext(\exquants^+_{(1-\epsilon) \budget}) \geq \ccext((1 - \epsilon) \exquants^+_{ \budget}) \geq (1-\epsilon) \ccext( \exquants^+_{ \budget}). $$ Since $\ccext( \exquants^+_{ \budget})$ is an upper bound on the performance of the optimal ex ante mechanism by Lemma~\ref{l:ub}, the ex ante posted pricing mechanism defined for each agent
by $\price_i = \disti^{-1}(\exquanti)$ is a $(1 - 1/e)(1 - \epsilon)$ approximation to the optimal mechanism.

We now consider the posted pricing mechanism defined by $\prices$ that is no longer ex ante. Since the budget has been lowered by a factor $1 - \epsilon$, each agent is offered her price with probability at least $\prob[R \sim \exquants]{\sum_{i \in R} \critcosti \leq (1-1/k) \budget}$ by Lemma~\ref{l:budget-offer-relation}, regardless of the ordering $\sigma$ of agents. By Theorem~\ref{t:CR-knapsack}, this probability is at least $1- e^{-
  \epsilon^2 ( 1- \epsilon) k / 12}$. Therefore, by Lemma~\ref{l:CR-main}, the expected value of this mechanism is at least $(1- e^{-
  \epsilon^2 ( 1- \epsilon) k / 12}) \mlext(\critquants)$ and this mechanism is a $(1 - \epsilon) (1 - 1/e) (1- e^{-
  \epsilon^2 ( 1- \epsilon) k / 12})$ approximation to the optimal mechanism in the case of submodular value functions. In the case of additive functions, there is no loss from the multilinear extension to the concave closure, so the mechanism is a $(1 - \epsilon) (1- e^{-
  \epsilon^2 ( 1- \epsilon) k / 12})$ approximation.
\end{proof}

Note that as the size of the market $k$ grows to infinity, this
approximation ratio approaches $1 - 1/e$. Also note that this
mechanism requires the market to be at least $4$-large. Using another result from \citet{VCZ-11} and a similar analysis to the one from this section, a $(1 - 1
/e)/8$ posted pricing mechanism can easily be obtained for any market size. This posted pricing attains
its performance guarantee when agents with cost at least $\budget/4$
arrive before all others, but otherwise the order is oblivious.


\section{Additive Value and Sequential Posted Pricing}
\label{s:additive}

In this section we give improved bounds for sequential posted pricing,
i.e., where the mechanism orders the agents, and when the value
function is additive, i.e., $\val(S) = \sum_{i \in S} \vali$.  In
particular, we analyze the sequential posted pricing $(\critcosts,
\orders)$ with $\critcosti = \disti^{-1}(\exquanti)$ from the solution
to the multilinear extension program~\eqref{eq:multilinear-prog} with
the full budget $\budget$ and the ordering $\orders$ by decreasing
bang-per-buck, i.e., $\vali / \critcosti$ for agent $i$.

Our results in this section are based on the analysis of the
correlation gap of fractional and integral-knapsack set functions (to
be defined subsequently).  The fractional-knapsack set function is a
submodular function, so a correlation gap of $1 - 1/e$ can be directly
obtained (Theorem~\ref{t:cg-submodular}). In this section, we improve
this bound to $1 - 1 / \sqrt{2 \pi k }$ for $k$-large markets, i.e.,
with $k = \budget / \max_i \pricei$.  From this bound we observe that
the correlation gap for fractional-knapsack in large market is
asymptotically one.  We show that the integral-knapsack correlation
gap is nearly the same.  Following the approach of \citet{yan-11}, the
factor by which sequential posted pricing approximates the ex ante
relaxation is equal to the integral-knapsack correlation gap.


\begin{definition}
\label{d:fractional-knapsack-value}
The {\em fractional-knapsack} set function corresponding to additive
set function $\val(S) = \sum_{i\in S} \vali$, sizes $\prices$, and
capacity $\budget$ is denoted $\valB(S)$ and equals the maximum value
solution to the corresponding fractional-knapsack problem on elements
$S$.\footnote{This value is given by sorting the elements of $S$ by
  $\vali/\pricei$ and admitting them greedily until the first element
  that does not fit with the remaining capacity, that element is
  admitted fractionally (providing a fraction of its value).}  The {\em integral-knapsack} set function can be defined analogously to
the fractional one, but it cannot add elements fractionally.
\end{definition}

Most of this section analyzes the ratio of the independent value of fractional-knapsack to the correlated value of $\val(\cdot)$ (see Definition~\ref{d:submodular-extensions} for the definition of independent and correlated values) in the case where the budget constraint is met ex ante, i.e., $\expect[S
     \sim \exquants]{\valB(S)} / \expect[S
     \sim \setdist]{\val(S)}$ when $\sum_i \pricei \exquanti \leq \budget$. We then show that this ratio is equal to the approximation ratio of the sequential posted pricing mechanism. Finally, we use this ratio to bound the integral, and fractional, knapsack correlation gap.

The main idea to derive a bound on this ratio is to show that it is minimized when all agents have equal cost $\budget / k$, in which case, when the budget constraint is met ex ante, we can then apply the result from \citet{yan-11} for the correlation gap of the k-highest-value-elements set function.

\begin{lemma}\label{l:maxcBk}
For any additive value function $\val(\cdot)$ and budget $\budget$, over
marginal probabilities $\exquants$ and prices $\prices$ that (a) satisfy
the ex ante budget constraint, i.e., $ \sum_i \pricei\,\exquanti \leq \budget$,
and (b) satisfy the $k$-large market condition, i.e., $\pricei \leq
\budget/k$, the ratio of the independent value of the fractional-knapsack and the correlated value of $\val(\cdot)$ is minimized when
$\pricei = \budget/k$ for all $i$.
\end{lemma}

\begin{proof} 
For the first part of the proof, we assume that $\vals = \prices$,
i.e., that the bang-per-buck is one for all elements. The last step of
the proof is to generalize this special case to any values. Observe
that with this assumption, $\valB(S) = \min(\budget, \sum_{j \in S}
\pricei[j])$.

Assume that there is some $\pricei$ such that $\pricei < \budget /
k$. We show that when $\vali = \pricei$, increasing $\pricei$ to any
$\pricei' > \pricei$ and decreasing $ \exquanti$ to $\exquanti' =
\pricei \exquanti / \pricei'$ preserves the correlated value while
only lowering the independent value. Let $\price_j' = \price_j$ and
$\exquant_j' = \exquant_j$ for $j \neq i$. The correlated value of
$\val(\cdot)$ is $\expect[S \sim \setdist]{\val(S)} = \sum_j \price_j
\exquant_j = \sum_j \price_j' \exquant_j'$ so it is preserved. Similarly, the ex ante budget constraint is still satisfied.

The argument for the independent value decreasing is the following.
Let $\valB^{\prime}(S)$ be defined similarly as $\valB(S)$, but where
agents have values and costs equal to $\prices'$.  Condition on the
subset of other agents $S$ who accept their prices and consider the
marginal contribution to the expected value of $\valB(\cdot)$ and
$\valB^{\prime}(\cdot)$ from agent $i$.  In the case that $C = \sum_{j
  \in S} \pricei[j] > \budget$, this contribution is zero for both
$\pricei$ and $\pricei'$.  When $C < \budget$,
these contributions are $\exquanti\min(\budget - C,\pricei)$ and
$\exquanti'\min(\budget - C,\pricei')$.  By the definition of
$\exquanti' = \pricei \exquanti/\pricei'$ and concavity of
$\min(\budget-C,\cdot)$, the former is greater than the latter.  This
inequality holds for all sets $S$, so removing the conditioning on $S$, it
holds in expectation and the independent value of fractional-knapsack is lowered.





It remains to extend this result to any $\vals$. Fix $\vals$ and assume without loss of generality that $\val_1 / \price_1  \geq \cdots \geq \val_n / \price_n$. Then the fractional-knapsack set function can be rewritten as $$\val_{\budget}(S) = \sum_{i \in N} (\vali /\pricei - \val_{i + 1}/ \price_{i+1}) \min(\budget, \sum_{j \in S \cap \{1, \dots, i\}} \price_j)$$ and the additive set function as $$\val(S) = \sum_{i \in N} (\vali /\pricei - \val_{i + 1}/ \price_{i+1})  (\sum_{j \in S \cap \{1, \dots, i\}} \price_j)$$ since these sums telescope.

 So the ratio of   independent value of $\valB(S)$ to the correlated value of $\val(S)$ is minimized when the ratios of the independent value of $\min(\budget, \sum_{j \in S \cap \{1, \dots, i\}} \price_j)$ to the correlated value of $\sum_{j \in S \cap \{1, \dots, i\}} \price_j$ are minimized for all $i$. We conclude by observing that $\min(\budget, \sum_{j \in S \cap \{1, \dots, i\}} \price_j)$ and $ \sum_{j \in S \cap \{1, \dots, i\}} \price_j$ are the fractional-knapsack set function and the additive set function when $\vali = \critcosti$ over ground set $\{1, \dots, i\}$, and that their ratio is minimized when $\pricei = \budget / k$ for all agents $i$.
\end{proof}

Next, we use the result from \citet{yan-11} to bound the ratio of the independent value of fractional-knapsack to the correlated value of $\val(\cdot)$.

\begin{lemma}
\label{l:fractional-knapsack-additive-relaxation-cg}
For any distribution over sets $\setdist$ with marginal probabilities
$\exquants$ satisfying the ex ante budget constraint, i.e., $\sum_i \pricei\,\exquanti \leq \budget$, the ratio of the independent value of fractional-knapsack to the correlated value of $\val(\cdot)$ is at least $1-1/\sqrt{2 \pi k}$ when the market is $k$-large.
\end{lemma}
\begin{proof}
Consider the case where each agent $i$ has cost $\pricei = \budget / k$ and assume that the ex ante budget constraint is satisfied, so $\sum_i \exquanti \leq k$.   Since any set of size at most $k$ is feasible and since $\sum_i \exquanti \leq k$, there is a distribution such that the budget constraint is always met ex post. Therefore, the correlated value of $\val(\cdot)$ is equal to the correlated value of fractional-knapasck. The ratio of the independent value of fractional-knapsack to the correlated value of $\val(\cdot)$ is thus equal to the correlation gap of fractional-knapsack. Since all agents have cost $\budget / k$, the fractional-knapsack set function is equal to the k-highest-value-elements set function. By Theorem~\ref{t:cg-multi-unit}, the ratio of  the independent value of fractional-knapsack to the correlated value of $\val(\cdot)$ is therefore $1-1/\sqrt{2 \pi k}$.

By Lemma~\ref{l:maxcBk}, the ratio of the independent value of fractional-knapsack to the correlated value of $\val(\cdot)$ when the ex ante budget constraint is satisfied is minimized when all agents have cost $\budget / k$, so this ratio is at least $1-1/\sqrt{2 \pi k}$.
\end{proof}

We now prove the main theorem of this section which relates the
approximation factor of sequential posted pricing (with ex post budget
feasibility) to the optimal mechanism with ex ante budget feasibility.

\begin{theorem}
\label{t:pp-additive}
The sequential posted pricing mechanism $(\exquants,\orders)$, where $\exquants$ is the solution to the multilinear extension
program~\eqref{eq:multilinear-prog} and where the order $\orders$ is decreasing in
$\frac{\vali}{\pricei}$,  is a $(1-1/\sqrt{2 \pi
  k})(1 - 1/k)$ approximation to the optimal mechanism in the case of regular cost distributions.
\end{theorem}

\begin{proof}
Denote $\exquants$ the optimal solution to the multilinear extension
program~\eqref{eq:multilinear-prog}.  For additive value functions,
linearity of expectation implies that the multilinear extension is
equal to the concave closure and the optima of the multilinear
extension program~\eqref{eq:multilinear-prog} and concave closure
program~\eqref{eq:relaxed-concave-bayes-opt-prog} are the same.  Their
performance upper bounds that of the optimal mechanism that satisfies
ex post budget feasibility by Lemma~\ref{l:ub}. The objective value of these programs with optimal solution $\exquants$ is $\sum_i \vali \exquanti$, which is equal to the correlated
value of the additive set function $\val(\cdot)$ on distributions with
marginals $\exquants$.  So by Lemma~\ref{l:fractional-knapsack-additive-relaxation-cg}, the ratio of the independent value of fractional-knapsack to the upper bound of the optimal mechanism is at least $1-1/\sqrt{2 \pi k}$

The random set of agents who accept their offer in the sequential posted pricing is equal to the set of agents who are admitted by the fractional-knapsack set function on an independent random set of agents with marginals $\exquants$, without including the fractional agent. The loss from this fractional agent is at most a factor $1 - 1/k$. This posted pricing mechanism therefore has an approximation ratio of $(1-1/\sqrt{2 \pi k})(1-1/k)$.
\end{proof}

As a corollary of Lemma~\ref{l:fractional-knapsack-additive-relaxation-cg}, we get new correlation gap results for the fractional, and integral, knapsack set functions. 

\begin{theorem}
\label{t:cg}
The correlation gaps of fractional-knapsack and integral-knapsack are at least $1-1/\sqrt{2 \pi k}$ and $(1-1/\sqrt{2 \pi k})(1-1/k)$ respectively, in a $k$-large market.
\end{theorem}
\begin{proof}
We first show the correlation gap of fractional-knapsack, the correlation gap of integral-knapsack will then follow easily. We start by showing that the correlation gap is minimized when the budget constraint is satisfied. Then, we upper bound the fractional-knapsack correlated value by the correlated value of $\val(\cdot)$. Finally, we apply Lemma~\ref{l:fractional-knapsack-additive-relaxation-cg}.

We claim that the correlation gap of fractional-knapsack is minimized when the budget constraint is satisfied. Observe that if the budget constraint is not satisfied, then it is possible to decrease some $\exquanti$ such that the correlated value of fractional-knapsack remains the same. Since decreasing some $\exquanti$ only decreases the independent value of fractional-knapsack, the ratio of the independent value to the correlated value also decreases.

Clearly, the fractional-knapsack correlated value is upper bounded by the correlated value of $\val(\cdot)$. Therefore, the correlation gap of fractional-knapsack is at least the ratio of the independent value of fractional-knapsack to the correlated value of $\val(\cdot)$ when the budget constraint is satisfied, so at least $1-1/\sqrt{2 \pi k}$ by Lemma~\ref{l:fractional-knapsack-additive-relaxation-cg}.

Finally, observe that the correlated value of fractional-knapsack upper bounds the correlated value of integral- knapsack and that the independent value of integral-knapsack is a $1 - 1/k$ approximation to the independent value of fractional-knapsack. Therefore, the correlation gap of integral-knapsack is at least $(1-1/\sqrt{2 \pi k})(1-1/k)$.
\end{proof}

\paragraph{Comparison of Sequential and Oblivious posted pricing.} 
We now compare the approximation ratio for additive value functions achieved using the sequential
posted pricing mechanism with the bang per buck order, $(1-1/\sqrt{2 \pi k})(1 - 1/k)$, and using oblivious posted pricing where the budget is lowered, $(1 -
\epsilon) (1- e^{- \epsilon^2 ( 1- \epsilon) k / 12})$. Figure~2 shows that the approximation ratio with the sequential ordering approaches $1$ much faster than with the oblivious ordering as the size of the market increases. To obtain these results for oblivious posted
pricing, we numerically solved for the best $\epsilon$. We emphasize that we are comparing the theoretical bounds of these approaches, and not empirical performances.

\begin{figure}\label{f:comp}
\centering
  \includegraphics[scale = .5, bb=0 0 550 150]{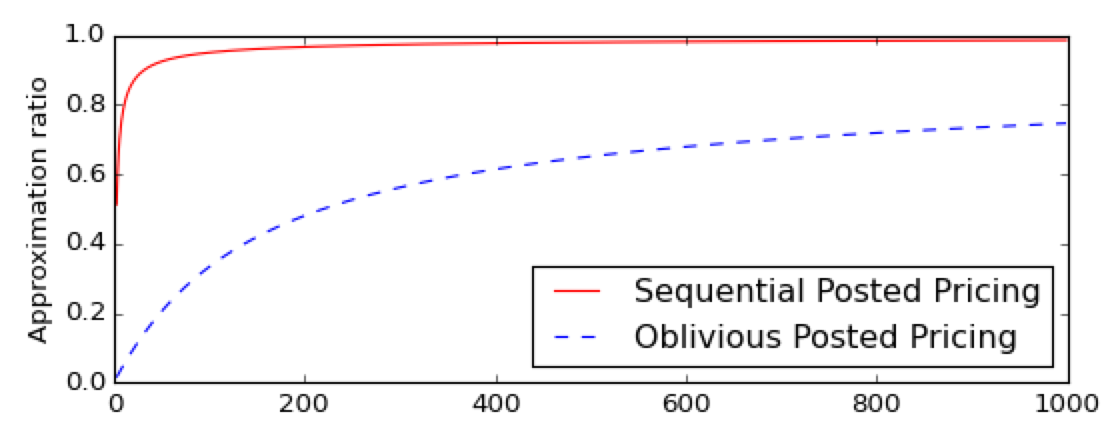}
\caption{Comparison of the approximation ratios obtained for additive value functions by the two different approaches. On the horizontal axis is $k$, the size of the market. }
\end{figure}

\section{Computing Prices}
\label{s:optimal}
In the two previous sections, we gave conditions under which optimal
prices from the multilinear extension
program~\eqref{eq:multilinear-prog} perform well when offered
sequentially or obliviously.  In this section, we consider the
computational problem of finding these prices. For submodular
value functions, we reduce the problem to the well-known greedy algorithm for submodular optimization.  For additive value functions, we use a simple method based
on the Lagrangian relaxation of the budget constraint.

\subsection{The Lagrangian Relaxation for Additive Value Functions}
Consider the case of additive value functions where the principal has
a value $\vali$ for each agent $i$ and the value function is
$\val(S) = \sum_{i \in S} \vali$. Recall the virtual surplus
program~\eqref{eq:bayes-opt-prog} from Section~\ref{s:relaxations}:
\begin{align}
\max_{\allocs}\ &\expect[\costs]{\val(\allocs(\costs))} \tag{\ref{eq:bayes-opt-prog}} \\
\notag    \text{s.t. } &\sum\nolimits_i \expect[\costs]{\virti(\costs)\,\alloci(\costs)} \leq \budget,
\\
\intertext{which can be rewritten for additive value functions as:}
\label{eq:relaxed-bayes-opt-prog-add}
\max_{\quants}\ &\sum\nolimits_i \expect[\costs]{ \vali \, \alloci(\costs)}\\
\notag    \text{s.t. } &\sum\nolimits_i \expect[\costs]{\virti(\costs)\,\alloci(\costs)} \leq \budget.
\end{align}
 
We show that the ex ante optimal mechanism can be found directly by taking the Lagrangian relaxation of the budget constraint (with parameter $\lagrange$) of the following Lagrangian program:

\begin{align}
\label{eq:relaxed-bayes-opt-prog-virt}
\max_{\allocs}\ & \lagrange \budget +  \sum\nolimits_i \expect[\costs]{(\vali - \lagrange \virti(\costi))\,\alloci(\costs)}.
\end{align}

For any Lagrangian parameter $\lagrange$, this objective can be
optimized by pointwise optimizing $\sum_i (\vali - \lagrange \virti(\costi))\,\alloci(\costs)$, a.k.a., the {\em Lagrangian virtual
  surplus}. This pointwise optimization picks all the agents such that  $\vali \geq \lagrange \virti(\costi)$.  If the virtual cost functions are monotone, i.e., in the
so-called {\em regular} case, then this optimization gives a
monotone allocation rule where an agent is picked whenever $\costi \leq  \virti^{-1}(\vali / \lagrange)$

Notice that as the Lagrangian parameter increases, the payments of
the agents, as represented by virtual costs, become more costly in the
objective of the lagrangian program~\eqref{eq:relaxed-bayes-opt-prog-virt}. Thus, the
expected payment of the mechanism is monotonically decreasing in the
Lagrangian parameter.  With $\lagrange = 0$ the Lagrangian virtual
surplus optimizer simply maximizes $\val(\allocs)$ and pays each
agent selected the maximum cost in the support of her distribution.
If this payment is under budget then it is optimal, otherwise, we can
increase $\lagrange$ until the budget constraint is satisfied.  For
example, with $\lagrange = \infty$ the empty set of agents is selected
and no payments are made.  The optimal mechanism is the one that meets
the budget constraint with equality.  In the case that the expected
payment is discontinuous then mixing between the least over-budget and
least under-budget mechanism is optimal.   For
further discussion of Lagrangian virtual surplus optimizers, see \citet{DHH-13}.

\begin{proposition}
\label{prop:bayes-opt}
The Lagrangian virtual surplus optimizer (or appropriate mixture
thereof) that meets the budget constraint with equality is the
Bayesian optimal ex ante budget feasible mechanism.
\end{proposition}

Lagrangian virtual surplus optimization suggests selecting an agent $i$
when her private cost $\costi$ is below $\virti^{-1}(\vali /
\lagrange)$.  The mechanism that achieves this outcome posts the price
of $\pricei = \virti^{-1}(\vali/\lagrange)$ to agent $i$.  Denote by
$\exquanti = \disti(\pricei)$ the probability that $i$ accepts the
price $\pricei$.  For the prices $\prices$, the total expected
payments are $\sum_i \pricei\,\exquanti$.  When the virtual cost
functions are monotone and strictly increasing, there is a Lagrangian
parameter for which the budget constraint is met with equality, i.e.,
with $\sum_i \pricei\,\exquanti = \budget$. The optimal ex ante mechanism is
therefore the posted price mechanism that posts $\pricei$ to each
agent $i$ for the Lagrangian parameter $\lagrange$ that satisfies
$\sum_i \pricei\,\exquanti = \budget$. Note that such a Lagrangian parameter
$\lagrange$ can be arbitrarily well approximated since $\sum_i
\pricei\,\exquanti$ is decreasing as a function of $\lagrange$.

\begin{example} 
 Consider $n$ agents with costs drawn uniformly and i.i.d.\@ from
 $[0,1]$ and uniform additive value function $\vali = 1$ for all $i$,
 i.e., the cardinality function.  The virtual cost function is
 $\virt(\cost) = \cost + \frac{\dist(\cost)}{\dens(\cost)} = 2\cost$.
 The Lagrangian parameter $\lagrange = \frac{1}{2} \sqrt{n/\budget}$
 induces a uniform posted price of $\price = \sqrt{\budget/n}$ which
 is accepted with probability $\exquant = \sqrt{\budget/n}$ for an
 expected payment of $\budget/n$.  Summing over all $n$ agents,
 the budget is balanced ex ante.
\end{example}

\subsection{A Reduction to the Greedy Algorithm for Submodular Optimization}

For general submodular value functions we reduce the optimization of
the multilinear extension program~\eqref{eq:multilinear-prog},
restated below, to the problem of optimizing a submodular function
subject to a cardinality constraint. This problem of optimizing a submodular function under cardinality, knapsack, or matroid constraints is well
studied and the {\em greedy algorithm} gives a
$1-1/e$ approximation for knapsack and cardinality constraints; see \citet{NWF-78}, \citet{KMN-99}, and
\citet{SVI-04}.
\begin{align*}
\max_{\quants}\ &\mlext(\quants) \tag{\ref{eq:multilinear-prog}} \\ 
\notag    \text{s.t. } &\sum\nolimits_i \quanti \disti^{-1}(\quanti) \leq \budget.
\end{align*} 
Define the \emph{cost curve} of agent $i$ to be the expected payment to agent $i$, i.e., $\quanti \disti^{-1}(\quanti)$ in our case.
The main difference between the multilinear extension program
\eqref{eq:multilinear-prog} and the knapsack setting considered in the
literature is that the cost curves in the knapsack setting are linear in $\quanti$. Our reduction to the greedy algorithm is the following. We divide each agent $i$, called a big agent, in cost space into $\numSmall$ discrete agents $i_j$ of equal cost, called the small agents. An agent $i_j$ corresponds to the $j$th increase of $\quanti$, starting from $\quanti = 0$, that has cost $\budget / \numSmall$.
  We set
$1 / \numSmall$ as a fraction of the total budget $\budget$ which fixes the
number of steps in the algorithm to be $\numSmall$. With large $\numSmall$, the reduction becomes a finer discretization.  

Before formally describing the reduction, we introduce some notation. For each $i$ and $j$,  let $\delta_{ij}$  be the $j$th increase in $\quanti$, starting from $\quanti= 0$,
  that has cost $\budget / \numSmall$, i.e., $\delta_{ij}$ satisfying $ \budget / \numSmall = \disti^{-1}(\sum_{k \leq j} \delta_{ik}) \cdot (\sum_{k \leq j} \delta_{ik})  - \disti^{-1}(\sum_{k < j} \delta_{ik}) \cdot (\sum_{k < j} \delta_{ik})$. Given a set $S$ of small agents, the continuous solution  corresponding to $S$ is $\quants(S)$ with $\quanti(S) = \sum_{j : i_j \in S} \delta_{ij}$.

\paragraph{The reduction.}
\begin{enumerate}
\item For each agent $i$, create $ \numSmall$ small agents $i_j$ where
  $1 \leq j \leq \numSmall$ so that the reduced instance has $
  \numSmall n$ agents.
\item For each small agent $i_j$, its cost is $ \budget / \numSmall$.
\item For each small agent $i_j$, its marginal contribution
  $\mlext_S(i_j)$ in value to a set $S$ is the marginal contribution
  of increasing the fraction of agent $i$ corresponding to $S$ by
  $\delta_{ij}$, i.e., $\mlext(\quants') - \mlext(\quants(S))$ where
  $\quanti' = \quanti(S) + \delta_{ij}$ and $\quant_j' = \quant_j(S)$
  for $j \neq i$.
\end{enumerate}

We show that the solution to the reduced problem that we obtained with the greedy algorithm for cardinality constraint corresponds to a solution for the multilinear extension program~\eqref{eq:multilinear-prog} that is a
$1 - 1/e - o(1)$ approximation, almost matching the performance of the greedy
algorithm for knapsack constraint with integral agents and linear cost curves. We start by showing that if a solution is feasible in the reduced problem, then the continuous solution corresponding to it is a feasible solution to the multilinear extension program~\eqref{eq:multilinear-prog}. Then, with access to exact values of the increases $\delta_{ij}$ and of the marginal contributions $\mlext_S(i_j)$, the approximation ratio is
$1 - 1/e - o(1)$.  Finally, we show that it is possible to approximate $\delta_{ij}$ and $\mlext_S(i_j)$ with estimates that cause an additional loss of $o(1)$ to the approximation ratio.

\paragraph{From a set of small agents to a continuous solution for the big agents.}

Previously, we defined a distribution to be regular if the virtual cost function is monotonically increasing. An alternate definition is that a distribution $\dist$ is regular if the cost curve $\quant \cdot \dist^{-1}(\quant)$ is convex. This definition is the analogue to the revenue curve being concave for regular distributions when the agents are buyers, and not sellers, from \citet{BR-89}.

Recall that given a set $S$ of small agents, the continuous solution  corresponding to $S$ is $\quants(S)$ with $\quanti(S) = \sum_{j : i_j \in S} \delta_{ij}$ and that $\delta_{ij}$ is the $j$th increase in $\quanti$ that has cost $\budget / \numSmall$. Therefore, given a set $S$ of small agents of size at most $\numSmall$ such that for any $\delta_{ij} \in S$, $\delta_{ik} \in S$ for all $k < j$, then $\quants(S)$ has cost at most $\budget$. The condition that if $\delta_{ij} \in S$, then $\delta_{ik} \in S$ for all $k < j$, is equivalent to the condition that greedy always picks small agents corresponding to lower quantiles before small agents corresponding to higher quantiles, which we show formally.

\begin{lemma}
Given two small agents $i_k$ and $i_j$ such that $k < j$, the greedy algorithm with a cardinality constraint picks $i_k$ before $i_j$ for regular distributions $\disti$.
\end{lemma}
\begin{proof}
Since all small agents have equal cost, we need to show that $i_k$ has a larger marginal contribution than $i_j$ to any set $S$ of small agents such that $i_k, i_j \not \in S$. Since $\mlext(\cdot)$ is monotone, it suffices to show that $\delta_{ik} > \delta_{ij}$. In quantile space, the cost of increasing some quantile $\quanti$ by a fix amount is increasing in $\quanti$ since $\quanti \cdot \dist^{-1}(\quanti)$ is convex by definition of regular distributions. Therefore, in cost space, the increase in quantile $\delta_i$ that is obtained by increasing the cost curve by a fix amount is decreasing, so $\delta_{ik} > \delta_{ij}$.
\end{proof}

The case of irregular distributions is considered in Section~\ref{s:irregular}.

\paragraph{With exact values of $\delta_{ij}$ and $\mlext_S(i_j)$.} We consider the case where the exact values of the increases in $\quants$ and marginal contributions
are given by an oracle.  We show that finding a good solution to this reduced problem with small agents gives us a good solution to the problem with big agents.

 \begin{lemma}\label{l:discretization} The optimal solution $\opti{S}$ to the reduced problem  satisfies $ \mlext(\quants(\opti{S})) \geq (1 - o(1)) \mlext(\exquants)$ where $\exquants$ is the optimal solution to the multilinear extension program~\eqref{eq:multilinear-prog}.
  \end{lemma}
  \begin{proof} We pick the step size to be $m = n^2$. The proof shows that there exists a set $S$ that is close to a feasible solution in the reduced problem and such that $\quants(S)$ is a better solution than $\exquants$. Let $S$ be the set of small agents such that $\quants(S)$ is maximized subject to $\quants(S) \leq \exquants$. Define $S^{+1}$ to be the set containing all small agents in $S$ and one additional small agent for each big agent $i$. Observe that $\mlext(\quants(S^{+1}))  \geq \mlext( \exquants)$ since $\mlext(\cdot)$ is non-decreasing. So there is a feasible solution to the discretized problem such that if we add one small agent for each big agent $i$, then we obtain a better solution than the optimal solution to the original problem. 
 
Greedily remove agents by minimal marginal contribution from $S^{+1}$ until we get a feasible solution $S$. The number of small agents who need to be removed is $n$ since $S$ is feasible. Since $S$ contains $n^2$ small agents, by the greediness and the fact  $\mlext(\cdot)$ is concave along any line of positive direction,  $(1 + 1/n) \mlext(\quants(S)) \geq \mlext(\quants(S^{+1}))$.

Therefore,
\begin{align*}
(1 + o(1)) \mlext(\quants(\opti{S})) \geq (1 + o(1)) \mlext(\quants(S)) \geq  \mlext(\quants(S^{+1})) \geq \mlext( \exquants)
\end{align*}
\end{proof}

Next, we show that the reduced problem can be optimized.

\begin{lemma}\label{l:greedyreduced} Let $S$ be the set returned by the greedy algorithm for submodular functions under a cardinality constraint on the reduced problem, then $ \mlext(\quants(S)) \geq (1 - 1/e) \mlext(\quants(S^*))$ where $S^*$ is the optimal solution to the reduced problem.
\end{lemma}
\begin{proof}
 Observe that the objective function in the reduced problem is a submodular function. This follows directly from the concavity of $\mlext(\cdot)$ along any positive line of direction. In addition, since all small agents have cost $\budget / \numSmall$, the constraint is a cardinality constraint. Since the greedy algorithm for submodular functions under a cardinality constraint is a $1 -1 /e$ approximation for submodular functions, we get the desired result.
\end{proof}

We now have the tools to show that if we had an oracle for the increases and marginal contributions, the greedy algorithm on the reduced instance would give us a  $1 - 1/e - o(1)$ approximation.

\begin{lemma} \label{l:greedyExact} Let $S$ be the output of the greedy algorithm on the reduced instance, where exact values of $\delta_{ij}$ and $\mlext_S(i_j)$ are given by an oracle at each iteration, then $\mlext(\quants(S)) \geq (1 - 1/e - o(1)) \mlext(\exquants)$, where $\exquants$ is the optimal solution to the multilinear extension program~\eqref{eq:multilinear-prog}.
\end{lemma}
\begin{proof}
We combine the results from the discretization that causes a $o(1)$ loss with the greediness of the algorithm that is a $1 - 1/e$ approximation to obtain the desired result.

 By  Lemma~\ref{l:greedyreduced} and  Lemma~\ref{l:discretization},
$$ \mlext(\quants(S))  \geq (1 - 1/e) \mlext(\quants(S^*)) \geq (1 - 1/e - o(1)) \mlext(\critquants)$$ where $S^*$ is the optimal solution to the reduced problem.
\end{proof}

\paragraph{With estimates of $\delta_{ij}$ and $\mlext_S(i_j)$.} We now show that we can use the greedy algorithm with estimates of the increases and the marginal contributions, that we can compute. Let $\noisy{\quants}(S)$ be defined similarly to $\quants(S)$ but with estimates $\noisy{\delta}_{i_j}$.  The first lemma shows that the value of the optimal solution to the reduced problem has almost the same value as when the increases $\delta_{ij}$ are estimated. The second lemma extends Lemma~\ref{l:greedyreduced} to the case where greedy is run with estimated marginal contributions $\noisy{\mlext}_S(i_j)$ and any $\noisy{\delta}_{i_j}$. We defer the proofs of these two lemmas to the appendix.

\begin{lemma} \label{l:noisygreedy} Let $S^*$ be the optimal solution to the reduced problem with exact value of $\delta_{ij}$ and $\mlext_S(i_j)$,  then $  \mlext(\noisy{\quants}(S^*)) \geq (1 - o(1)) \mlext(\quants(S^*))$.
\end{lemma}

\begin{lemma} \label{l:ngreedyreduced}
Let $\noisy{S}$ be the set returned by the greedy algorithm on the reduced problem with estimates $\noisy{\delta}_{i_j}$ and $\noisy{\mlext}_S(i_j)$, then $ \mlext(\noisy{\quants}(\noisy{S})) \geq (1 - 1/e - o(1)) \mlext(\noisy{\quants}(S^*))$ w.h.p., where $S^*$ is the optimal solution to the reduced problem with estimates $\noisy{\delta}_{i_j}$ and exact values $\mlext_S(i_j)$.
\end{lemma}

Combining the previous results, we obtain the main result of this section.

\begin{theorem} \label{thm:greedy} Let $\noisy{S}$ be the output by the greedy algorithm on the reduced instance with estimates of $\delta_{ij}$ and $\mlext_S(i_j)$, then $\mlext(\noisy{\quants}(\noisy{S})) \geq (1 - 1/e - o(1)) \mlext(\exquants)$ w.h.p., where $\exquants$ is the optimal solution to the multilinear extension program~\eqref{eq:multilinear-prog}.
\end{theorem}
\begin{proof}
This proof follows similarly to the one for Lemma~\ref{l:greedyExact}, the difference is that this proof adds the loss from the estimates.

 By Lemma~\ref{l:noisygreedy} and Lemma~\ref{l:ngreedyreduced}, 
$$ \mlext(\noisy{\quants}(\noisy{S})) \geq (1 - 1/e - o(1)) \mlext(\noisy{\quants}(S^*)) \geq (1 - 1/e - o(1)) \mlext(\quants(S^*))$$ where $S^*$ is the optimal solution to the reduced problem. Using Lemma~\ref{l:discretization} that connects the discretized reduced instance to the original continuous problem, we conclude that $$\mlext(\noisy{\quants}(\noisy{S})) \geq (1 - 1/e - o(1)) \mlext(\quants(S^*)) \geq (1 - 1/e - o(1)) \mlext(\critquants).$$
\end{proof}

Note that in the case of additive value functions, the greedy algorithm is optimal when the optimization is subject to a cardinality constraint and the marginal contributions can be computed exactly. We therefore get the following result.

\begin{lemma} \label{l:greedy} Assume $\val(\cdot)$ is an additive value function. Let $S$ be the set returned by the greedy algorithm on the reduced problem with estimates $\noisy{\delta}_{i_j}$, then $ \mlext(\noisy{\quants}(S)) \geq (1- o(1)) \mlext(\exquants)$ w.h.p., where $\exquants$ is the optimal solution to the multilinear extension program~\eqref{eq:multilinear-prog}.
\end{lemma}

Therefore, all the results in previous sections suffer an extra $1 - 1/e - o(1)$ factor in the general case of submodular value function and an extra $1 - o(1)$ factor in the case of additive value function that are due to computational constraints.

\section{Symmetric Costs and Values}
\label{s:sym}
In this section we study symmetric environments where both the distribution of costs and the value function are symmetric. A submodular value function is symmetric if the value of a set only depends on the cardinality of that set, i.e., $\val(S) = g(|S|)$ for some function $g(\cdot)$. In this setting, we obtain an oblivious posted pricing that achieves an approximation ratio of $(1 - 1/\sqrt{2 \pi k })(1 - 1 / k)$ where $k$ is the size of the market, which is identical to the approximation obtained in the additive case with sequential posted pricing. We assume that the distribution of costs is regular.

The following technicalities are used for this section only. We overload the notation and denote by $\val(\cdot) : \mathbb{R}_+ \rightarrow \mathbb{R}_+$ the concave hull of the points $\{(i, v(S_i))\}_{i=0}^{n}$ where $S_i$ is any set of size $i$. The posted prices in this section are symmetric and are defined by a single price $\price$, i.e., $\prices = (\price, \cdots, \price)$ and $\exquants = (\exquant, \cdots, \exquant)$. Note that the market size $k$ in such a symmetric setting is $k = \budget / \price$.

We start with two lemmas that highlight symmetric properties of the optimal solution to the concave closure program in this symmetric setting.

\begin{lemma} 
\label{l:prices-symmetric}
For symmetric submodular value function $\val(\cdot)$ and symmetric distributions of costs,   the optimal solution $\exquants$ to the concave closure program~\eqref{eq:relaxed-concave-bayes-opt-prog} is symmetric, i.e., $\exquanti^+ = \exquant_j^+$ for all $i,j$.
\end{lemma}
\begin{proof} By the concavity of the concave closure and the convexity of cost curves (since the distribution of costs is regular), the program we wish to optimize is symmetric and convex, so the optimal solution is symmetric. 
\end{proof}

\begin{lemma}
\label{l:sizeSym}
For symmetric monotone submodular value function $\val(\cdot)$ and symmetric distributions of costs, there exists a distribution $\setdist$ over sets of agents with marginals $\exquants^+ = (\exquant^+, \cdots, \exquant^+)$ such that $\expect[S \sim \setdist]{\val(S)} = \ccext(\exquants^+)$ and such that all sets $S$ and $T$ that can be drawn from $\setdist$ have size either $\lfloor k \rfloor$ or $\lceil k \rceil$.
\end{lemma}
\begin{proof}
First, note that $\budget = \price \cdot n \cdot \exquant^+$ since $\exquants^+$ is the optimal solution to the concave closure program and since $\val(\cdot)$ is monotone, which implies that $k = n \cdot \exquant^+$ since $k = \budget / \price$.

  The expected value of a set of expected size $n \cdot \exquant^+$ drawn from a distribution is at most $\val(n \cdot \exquant^+)$ by the definition of concave hull. By taking a distribution $\setdist$ that is a mixture of sets of size $\lfloor n \cdot \exquant^+ \rfloor = \lfloor k \rfloor$ and $\lceil n \cdot \exquant^+ \rceil = \lceil k \rceil$ such that the marginals are $\exquant^+$, the expected value of a set drawn from $\setdist$ is $\val(n \cdot \exquant^+)$ since $\val(S)$ is submodular. Combining the two previous observations, $\expect[S \sim \setdist]{\val(S)} = \ccext(\exquants^+)$  since the concave closure is the maximum expected value over distributions with some marginals $\exquants$. 
\end{proof}

 Given quantiles $\exquants = (\exquant, \cdots, \exquant) $, the value of the concave closure $\ccext(\exquants)$ can be computed easily by Lemma~\ref{l:sizeSym} and symmetricity. The concave closure program can therefore be approximated arbitrarily well and efficiently when there is symmetry, by using binary search to get arbitrarily close to the optimal quantile $\exquant$.  Our approach for obtaining the desired approximation is to construct an additive function that lower bounds the symmetric submodular function on feasible sets and that upper bounds it otherwise.

\begin{theorem}
\label{t:pp-symmetric}
In the case of symmetric monotone submodular value functions and symmetric regular cost distributions, the oblivious posted pricing $\prices = (\price, \cdots, \price)$ with $\critcost = \dist^{-1}(\exquant^+)$  is an $(1 - 1 / \sqrt{2 \pi  k}) (1 - 1 / k)$ approximation to the optimal ex ante mechanism, where $\exquants^+ = (\exquant^+, \cdots, \exquant^+)$ is the optimal solution to the concave closure program~\eqref{eq:relaxed-concave-bayes-opt-prog} and $k$ is the size of the market.
\end{theorem} 

\begin{proof}
By Lemma~\ref{l:sizeSym}, there exists a distribution $\setdist$ over sets of agents with marginals $\exquant^+$ such that $\expect[S \sim \setdist]{\val(S)} = \ccext(\exquants^+)$ and such that sets drawn from $\setdist$ have size $\lfloor k \rfloor$ or $\lceil k \rceil$. We consider the additive value function $\val^{add}(\cdot)$ defined as follow: 
$$\val^{add}(S) = |S| \frac{\val(\lfloor k \rfloor)}{\lfloor k \rfloor}$$ and overload the notation for  $\val^{add}(\cdot)$ similarly as for $\val(\cdot)$. We make the following observations about $\val^{add}(\cdot)$:

\begin{itemize}
\item $\val^{add}(i) \leq \val(i)$ for $i \leq \lfloor k \rfloor$ and $\val^{add}(i) \geq \val(i)$ otherwise, by submodularity.
\item $\expect[S \sim \setdist]{\val^{add}(S)} \geq  \expect[S \sim \setdist]{\val(S)}$, since $\val^{add}(\lceil k \rceil) \geq  \val(\lceil k \rceil)$ and $\val^{add}(\lfloor k \rfloor) = \val(\lfloor k \rfloor)$.
\item $\val(\cdot)$ is an additive set function with values $\vali = \frac{1}{\lfloor k \rfloor}\val(\lfloor k \rfloor)$ for each element.
\end{itemize}

Since the feasible sets are sets of size at most $\lfloor k \rfloor$ and by the first observation on $\val^{add}(\cdot)$, the performance of the posted pricing mechanism is at least the independent integral knapsack value of $\val^{add}(\cdot)$. The independent integral knapsack value of $\val^{add}(\cdot)$ is at most a factor $(1 - 1/k)$ away from its independent fractional knapsack value,  $\expect[S \sim \exquants^+]{\val^{add}_{\budget}(S)} $. By Lemma~\ref{l:fractional-knapsack-additive-relaxation-cg} and the third observation on $\val^{add}(\cdot)$, $\expect[S \sim \exquants^+]{\val^{add}_{\budget}(S)} \geq (1 - \frac{1}{\sqrt{2 \pi k}}) \expect[S \sim \setdist]{\val^{add}(S)}$. By the second observation on $\val^{add}(\cdot)$,  $ \expect[S \sim \setdist]{\val^{add}(S)} \geq \ccext(\exquants^+)  $.  Since $\ccext(\exquants^+)$ is an upper bound on the optimal mechanism by Lemma \ref{l:ub}, we get the desired result.
\end{proof}

Note that in previous settings, we used the solution to the multilinear extension program to define the posted pricing mechanisms. In this setting, we used the solution to the concave closure program in order to take advantage of the concavity of the objective function for computational purposes. Finally, note that in the symmetric case, sequential posted pricing offers no advantage compared to oblivious posted pricing.

\section{Irregular Distributions}
\label{s:irregular}
In this section, we consider irregular distributions. Recall that a distribution $\dist$ is regular if the virtual cost function is increasing, or equivalently, if the cost curve $\quant \cdot \dist^{-1}(\quant)$ is convex. The ironing method introduced by \citet{mye-81} gives monotone ironed virtual costs and convex cost curves. With these convex cost curves, we construct \emph{randomized} posted pricing mechanisms that enjoy the same approximation ratios as the deterministic mechanisms, albeit with a generalized definition of the market size $k$ for randomized posted pricings. Additionally, in the case of additive objective functions, the sequential posted pricing is derandomized.

Denote the cost curve of agent $i$ by  $\costcurvei(\quanti) = \quanti \disti^{-1}(\quanti)$. \citet{BR-89} observed that the derivative of the cost curve with respect to quantile is equal to the virtual cost function, $\costcurvei'(\quanti) = \virti(\costi)$.  The ironing method constructs the convex hull $\costcurvehi(\quanti)$ of the cost curve $\costcurvei(\cdot)$. For $\quanti = \disti(\costi)$, the ironed virtual costs are $\ivirti(\costi) = \costcurvehi'(\quanti)$. By taking the convex hull of the cost curves, we have convex cost curves and monotone ironed virtual costs as desired. The next two lemmas show that expected payments $\costcurvehi(\exquanti)$ are feasible while serving each agent with probability $\exquanti$, and that no incentive compatible mechanism can do better.

\begin{lemma}[\citealp{mye-81}, \citealp{BR-89}]
\label{l:irr-neg} 
For any agent with cost drawn from distribution $\disti$ and any incentive compatible mechanism that selects agent $i$ with probability $\exquanti$, the expected payment to agent $i$ is at least $\costcurvehi(\exquanti)$.
 \end{lemma}
 
We give the proof of the following known lemma since it exhibits how to pick the prices and the probabilities of the randomized mechanisms.
 
 \begin{lemma} [\citealp{mye-81}]
 \label{l:irr-pos} 
 Expected payment  $\costcurvehi(\exquanti)$ while serving agent $i$ with probability $\exquanti$  is achievable using a randomized posted pricing with at most two prices.
 \end{lemma}

\begin{proof}
Fix a seller $i$ and an ex ante sale probability $\exquanti$. If $\exquanti = \costcurvei(\exquanti)$, then it suffices to post price $\disti^{-1}(\exquanti)$. Otherwise, let $a$ be the largest quantile smaller than $\exquanti$ such that $\costcurvehi(a) = \costcurvei(a)$. Similarly, let $b$ the smallest quantile larger than $\exquanti$ such that $\costcurvehi(b) = \costcurvei(b)$. The interval $[a,b]$ corresponds to the ironed interval in which $\exquanti$ falls in. By the definition of convex hull, we get
$$\costcurvehi(\exquanti) = (1 - \frac{\exquanti - a}{b- a}) \costcurvehi(a) + (1 - \frac{b- \exquanti}{b- a}) \costcurvehi(b) = (1 - \frac{\exquanti - a}{b- a}) \costcurvei(a) + (1 - \frac{b- \exquanti}{b- a}) \costcurvei(b).$$

Therefore, posting price $\disti^{-1}(a)$ with probability $1 - \frac{\exquanti - a}{b- a}$ and $\disti^{-1}(b)$ with probability $1 - \frac{b- \exquanti}{b- a}$ has expected payment $\costcurvehi(\exquanti)$ and the ex ante probability that seller $i$ accepts the price is $\exquanti$. 
\end{proof}

By Lemma~\ref{l:irr-neg} and Lemma~\ref{l:irr-pos}, the ex ante results also hold for the irregular case using randomized posted pricing. The following definition generalizes the notion of posted prices to allow for randomization.

\begin{definition} For a randomized posted pricing $\exquants$,
\begin{itemize}
\item Prices $\price_{i1}$ and $\price_{i2}$ with probabilities of picking each price are induced by $\exquanti$.
\item Randomly pick $\price_{i1}$ or $\price_{i2}$.
\item In the case of sequential posted pricing, set the ordering to be in decreasing order of bang-per-buck.
\end{itemize}
\end{definition}
\begin{definition} With randomized posted pricing, a market is $k$-large if $\budget / \price_{ij} \geq k$ for all agents $i$ and $j \in \{1,2\}$.
\end{definition}

\subsection{From Ex Ante to Ex Post with Additive Value Functions}

For the additive case, we first show that the ex post randomized posted pricing performs well and then derandomize the mechanism.

\begin{theorem}
\label{t:perf-irr}
The randomized sequential posted pricing mechanism $(\exquants, \orders(\cdot))$ that serve agents with probability $\exquants$, where $\exquants$ is the solution to the multilinear extension
program~\eqref{eq:multilinear-prog} and where the order $\orders(\cdot)$ is decreasing in
$\frac{\vali}{\pricei}$,  is a $(1-1/\sqrt{2 \pi
  k})(1 - 1/k)$ approximation to the optimal mechanism in a $k$-large market.
\end{theorem}
\begin{proof}
We show that the randomized sequential posted pricing performs better than a deterministic sequential posted pricing with the same ex ante performance and a market that is $k$-large. Consider a randomized agent $i$ who is offered $\pricei = \price_{i1}$ with probability $\rho$ and $\pricei =\price_{i2}$ otherwise. Remove agent $i$ and replace it with two deterministic agents $i1$ and $i2$ with value $\vali$, who are offered $\price_{i1}$ and $\price_{i2}$ and who accept their price with probability $\rho \disti(\price_{i1})$ and $(1-\rho) \disti(\price_{i2})$ respectively. Call this new posted pricing the deterministic instance and the original posted pricing the randomized instance. 

Both instances have the same ex ante performance since the expected total cost remains the same and since agent $i$ accepts his offer with probability equal to the sum of the probabilities that agents $i_1$ and $i_2$ accept their offer. Fix a set $S$ of agents who accept their offer that does not include $i$ and fix these offers. Notice that in both the randomized and deterministic instance, there is an expected increase in the total cost of $\price_{i1} \rho \disti(\price_{i1}) + \price_{i2} (1- \rho) \disti(\price_{i2})$ caused by agent $i$ to $S$. However, in the randomized instance, this increase in cost is either $\price_{i1}$ or $\price_{i2}$ and in the deterministic instance, this increase in cost can also be $\price_{i1} + \price_{i2}$. Since agents are ordered by decreasing bang-per-buck, the loss from agents that do not fit in the ex post budget constraint
is greater in the deterministic case. Therefore, the loss of the fractional knapsack value with respect to the ex ante performance of the mechanism is greater in the deterministic instance.

Now note that this argument can be repeated inductively until all the agents left are deterministic. So the approximation ratio obtained by the randomized mechanism is $(1-1/\sqrt{2 \pi
  k})(1 - 1/k)$, by combining Lemma~\ref{l:fractional-knapsack-additive-relaxation-cg} and the $1-1/k$ loss from dropping the fractional agent.
\end{proof}

We now show that the mechanism can be derandomized.

\begin{theorem}
Any sequential randomized posted pricing $(\exquants, \orders(\cdot))$ can be modified into a sequential deterministic posted pricing in the case of additive value functions.
\end{theorem}
\begin{proof}
The proof proceeds in two steps. The first reduces the number of randomized agents until there is one left by using properties of ironed intervals. The second step is to simply pick the best of the two prices that are offered to the last randomized agent.

Consider a randomized posted pricing $(\exquants, \orders(\cdot))$ with at least two agents $i$ and $j$ that are randomized.  The marginal cost per unit value of these two agents are $ \costcurvehi'(\exquanti) / \vali = \ivirti(\costi) / \vali$ and $\ivirt_j(\cost_j) / \val_j$. Without loss of generality, assume $\ivirti(\costi) / \vali \leq \ivirt_j(\cost_j) / \val_j$. Since both of these agents are randomized, $\exquanti$ and $\exquant_j$ are within ironed intervals and their ironed virtual costs are constants within these intervals. With no loss in the objective function, we can therefore increase $\exquanti$ and decrease $\exquant_j$ such that the budget still binds and such that either $\exquanti$ or $\exquant_j$ is at the extremity of the ironed interval it is in, and therefore not randomized anymore. This construction can be repeated until one randomized agent is left.

Consider a randomized posted pricing with a unique randomized agent $i$ who is offered $\pricei = \price_{i1}$ with probability $\rho$ and $\pricei =\price_{i2}$ otherwise. The proof of Theorem~\ref{t:perf-irr} shows that the ratio between the performance of the optimal mechanism and the expected fractional knapsack value is at least $1-1/\sqrt{2 \pi
  k}$. Agent $i$ is either offered $\price_{i1}$ or $\price_{i2}$, so by expectations, with at least one of these two offers, the previous ratio is at least $1-1/\sqrt{2 \pi
  k}$. Dropping the fractional agent and keeping the best price to offer to agent $i$, we therefore get a $(1-1/\sqrt{2 \pi
  k})(1 - 1/k)$ approximation for a deterministic mechanism.  \end{proof}

\begin{corollary}
Any sequential randomized posted pricing $(\exquants, \orders(\cdot))$ can be modified with high probability into a sequential deterministic posted pricing in the case of additive value functions with an additional $o(1)$ loss in polynomial time.
\end{corollary}
\begin{proof}
We need to compute which offered price between $\price_{i1}$ and $\price_{i2}$ performs better in terms of fractional knapsack value. Fractional-knapsack is a submodular function and  the multilinear extension of submodular functions can be approximated arbitrarily well by sampling using Chernoff bounds. Therefore, with high probability, it is possible to compare arbitrarily well the fractional knapsack value obtained with the two offered prices to agent $i$. 
\end{proof}

\subsection{From Ex Ante to Ex Post with Submodular Value Functions}

With submodular value functions, the analysis for the oblivious randomized posted pricing is identical as the analysis for the oblivious deterministic posted pricing. In Section~\ref{s:crs}, Theorem~\ref{t:CR-knapsack} shows that by lowering the budget by some small amount, we get that the sum of the costs does not exceed the budget with high probability. Note that this results does not only hold for deterministic agents but also for randomized agents since the payment $p_i$ to an agent $i$ only need to be bounded by $B/k$ and is not restricted to be either $0$ or $\pricei$. Therefore, the sum of the costs does not exceed the budget with high probability in the randomized case as well and the remaining of the analysis of section~\ref{s:crs} also holds.

\begin{theorem} 
For $\epsilon \in (0,1/2)$, if the randomized oblivious posted pricing $\exquants$, where $\exquants$
is the optimal solution to the multilinear
extension program~\eqref{eq:multilinear-prog} with budget $(1 -
\epsilon) \budget$,
satisfies $2/ \epsilon \leq k \leq B / \max_i \pricei$, then this
posted pricing mechanism is a $(1 - 1 /e)(1 - \epsilon) (1- e^{-
  \epsilon^2 ( 1- \epsilon) k / 12})$ approximation to the optimal
mechanism for submodular value functions and $(1 - \epsilon) (1- e^{-
  \epsilon^2 ( 1- \epsilon) k / 12})$ for additive value functions.
\end{theorem} 

\section{Conclusion}

We consider questions of budget feasibility in a Bayesian setting.  We
show that simple posted pricing mechanisms are ex post budget feasible
and approximate the Bayesian optimal mechanism.  Our analysis first
considers the ex ante relaxation where the budget constraint is
allowed to hold in expectation.  Good approximations are obtained when
this ex ante relaxation is optimized for a slightly reduced budget or
when the agents are ordered by bang-per-buck (value divided by offered
price).  The latter approach, in the case of additive value functions
when it applies, gives better bounds.

Another method for designing posted pricing mechanisms from the
literature comes from the generalized magician's problem from
\citet{ala-14}.  Unfortunately, this approach does not satisfy the
monotonicity property of Theorem~\ref{t:cr-main} needed to apply known
results that give a good approximation in the case of submodular
functions.  Thus, it is unclear whether this approach can be adapted
to budget feasibility questions.





\newpage
\bibliographystyle{apalike}
\bibliography{biblio}

\begin{thebibliography}{}

\bibitem[Agrawal et~al., 2010]{ADSY-10}
Agrawal, S., Ding, Y., Saberi, A., and Ye, Y. (2010).
\newblock Correlation robust stochastic optimization.
\newblock In {\em Proceedings of the twenty-first annual ACM-SIAM symposium on
  Discrete Algorithms}, pages 1087--1096. Society for Industrial and Applied
  Mathematics.

\bibitem[Alaei, 2014]{ala-14}
Alaei, S. (2014).
\newblock Bayesian combinatorial auctions: Expanding single buyer mechanisms to
  many buyers.
\newblock {\em SIAM Journal on Computing}, 43(2):930--972.

\bibitem[Anari et~al., 2014]{AGN-14}
Anari, N., Goel, G., and Nikzad, A. (2014).
\newblock Mechanism design for crowdsourcing: An optimal 1-1/e competitive
  budget-feasible mechanism for large markets.
\newblock In {\em Foundations of Computer Science (FOCS), 2014 IEEE 55th Annual
  Symposium on}, pages 266--275. IEEE.

\bibitem[Badanidiyuru et~al., 2012]{BKS-12}
Badanidiyuru, A., Kleinberg, R., and Singer, Y. (2012).
\newblock Learning on a budget: posted price mechanisms for online procurement.
\newblock In {\em Proceedings of the 13th ACM Conference on Electronic
  Commerce}, pages 128--145. ACM.

\bibitem[Bansal et~al., 2010]{BKNS-10}
Bansal, N., Korula, N., Nagarajan, V., and Srinivasan, A. (2010).
\newblock On k-column sparse packing programs.
\newblock In {\em Integer Programming and Combinatorial Optimization}, pages
  369--382. Springer.

\bibitem[Bei et~al., 2012]{BCGL-12}
Bei, X., Chen, N., Gravin, N., and Lu, P. (2012).
\newblock Budget feasible mechanism design: from prior-free to bayesian.
\newblock In {\em Proceedings of the forty-fourth annual ACM symposium on
  Theory of computing}, pages 449--458. ACM.

\bibitem[Bulow and Roberts, 1989]{BR-89}
Bulow, J. and Roberts, J. (1989).
\newblock The simple economics of optimal auctions.
\newblock {\em The Journal of Political Economy}, pages 1060--1090.

\bibitem[Calinescu et~al., 2007]{CCPV-07}
Calinescu, G., Chekuri, C., P{\'a}l, M., and Vondr{\'a}k, J. (2007).
\newblock Maximizing a submodular set function subject to a matroid constraint.
\newblock In {\em Integer programming and combinatorial optimization}, pages
  182--196. Springer.

\bibitem[Calinescu et~al., 2011]{CCMV-11}
Calinescu, G., Chekuri, C., P{\'a}l, M., and Vondr{\'a}k, J. (2011).
\newblock Maximizing a monotone submodular function subject to a matroid
  constraint.
\newblock {\em SIAM Journal on Computing}, 40(6):1740--1766.

\bibitem[Chawla et~al., 2010]{CHMS-10}
Chawla, S., Hartline, J.~D., Malec, D.~L., and Sivan, B. (2010).
\newblock Multi-parameter mechanism design and sequential posted pricing.
\newblock In {\em Proceedings of the forty-second ACM Symposium on Theory of
  Computing}, pages 311--320. ACM.

\bibitem[Chawla et~al., 2012]{CHS-12}
Chawla, S., Hartline, J.~D., and Sivan, B. (2012).
\newblock Optimal crowdsourcing contests.
\newblock In {\em Proceedings of the twenty-third annual ACM-SIAM symposium on
  discrete algorithms}, pages 856--868. SIAM.

\bibitem[Chen et~al., 2011]{CGL-11}
Chen, N., Gravin, N., and Lu, P. (2011).
\newblock On the approximability of budget feasible mechanisms.
\newblock In {\em Proceedings of the twenty-second annual ACM-SIAM symposium on
  Discrete Algorithms}, pages 685--699. SIAM.

\bibitem[Devanur et~al., 2013]{DHH-13}
Devanur, N.~R., Ha, B.~Q., and Hartline, J.~D. (2013).
\newblock Prior-free auctions for budgeted agents.
\newblock In {\em Proceedings of the Fourteenth ACM Conference on Electronic
  Commerce}, pages 287--304. ACM.

\bibitem[Ensthaler and Giebe, 2014]{EG-14}
Ensthaler, L. and Giebe, T. (2014).
\newblock Bayesian optimal knapsack procurement.
\newblock {\em European Journal of Operational Research}, 234(3):774--779.

\bibitem[Es{\"o} and Futo, 1999]{EF-99}
Es{\"o}, P. and Futo, G. (1999).
\newblock Auction design with a risk averse seller.
\newblock {\em Economics Letters}, 65(1):71--74.

\bibitem[Ho et~al., 2015]{HSSV-15}
Ho, C.-J., Slivkins, A., Suri, S., and Vaughan, J.~W. (2015).
\newblock Incentivizing high quality crowdwork.
\newblock In {\em Proceedings of the 24th International Conference on World
  Wide Web}, pages 419--429. International World Wide Web Conferences Steering
  Committee.

\bibitem[Immorlica et~al., 2015]{ISS-15}
Immorlica, N., Stoddard, G., and Syrgkanis, V. (2015).
\newblock Social status and badge design.
\newblock In {\em Proceedings of the 24th International Conference on World
  Wide Web}, pages 473--483. International World Wide Web Conferences Steering
  Committee.

\bibitem[Khuller et~al., 1999]{KMN-99}
Khuller, S., Moss, A., and Naor, J.~S. (1999).
\newblock The budgeted maximum coverage problem.
\newblock {\em Information Processing Letters}, 70(1):39--45.

\bibitem[Myerson, 1981]{mye-81}
Myerson, R. (1981).
\newblock Optimal auction design.
\newblock {\em Mathematics of Operations Research}, 6:58--73.

\bibitem[Myerson and Satterthwaite, 1983]{MS-83}
Myerson, R. and Satterthwaite, M. (1983).
\newblock Efficient mechanisms for bilaterial trade.
\newblock {\em Journal of Economic Theory}, 29:265--281.

\bibitem[Nemhauser et~al., 1978]{NWF-78}
Nemhauser, G.~L., Wolsey, L.~A., and Fisher, M.~L. (1978).
\newblock An analysis of approximations for maximizing submodular set
  functions—i.
\newblock {\em Mathematical Programming}, 14(1):265--294.

\bibitem[Singer, 2010]{sin-10}
Singer, Y. (2010).
\newblock Budget feasible mechanisms.
\newblock In {\em The 51st Annual IEEE Symposium on Foundations of Computer
  Science}, pages 765--774. IEEE.

\bibitem[Singer and Mittal, 2013]{SM-13}
Singer, Y. and Mittal, M. (2013).
\newblock Pricing mechanisms for crowdsourcing markets.
\newblock In {\em Proceedings of the 22nd international conference on World
  Wide Web}, pages 1157--1166. International World Wide Web Conferences
  Steering Committee.

\bibitem[Singla and Krause, 2013]{SK-13}
Singla, A. and Krause, A. (2013).
\newblock Truthful incentives in crowdsourcing tasks using regret minimization
  mechanisms.
\newblock In {\em Proceedings of the 22nd international conference on World
  Wide Web}, pages 1167--1178. International World Wide Web Conferences
  Steering Committee.

\bibitem[Sviridenko, 2004]{SVI-04}
Sviridenko, M. (2004).
\newblock A note on maximizing a submodular set function subject to a knapsack
  constraint.
\newblock {\em Operations Research Letters}, 32(1):41--43.

\bibitem[Vondr{\'a}k et~al., 2011]{VCZ-11}
Vondr{\'a}k, J., Chekuri, C., and Zenklusen, R. (2011).
\newblock Submodular function maximization via the multilinear relaxation and
  contention resolution schemes.
\newblock In {\em Proceedings of the forty-third annual ACM symposium on Theory
  of computing}, pages 783--792. ACM.

\bibitem[Yan, 2011]{yan-11}
Yan, Q. (2011).
\newblock Mechanism design via correlation gap.
\newblock In {\em Proceedings of the twenty-second annual ACM-SIAM symposium on
  Discrete Algorithms}, pages 710--719. SIAM.

\end{thebibliography}
\newpage 

\section*{APPENDIX}

\subsection*{Missing proofs from section \ref{s:optimal}}

\begin{proof}[Proof of Lemma~\ref{l:noisygreedy}] We need to find the increase satisfies $ \budget / \numSmall = \disti^{-1}(\sum_{k \leq j} \delta_{i_k}) \cdot (\sum_{k \leq j} \delta_{i_k})  - \disti^{-1}(\sum_{k < j} \delta_{i_k}) \cdot (\sum_{k < j} \delta_{i_k})$. To approximate it, we find $\noisy{\delta}_{i_j}$ such that $(1 - 1/n^3) \delta_{i_j} \leq \noisy{\delta}_{i_j} \leq \delta_i$, which can be done easily since the weight functions are increasing.

Recall that $\noisy{\quants}(S)$ is defined similarly to $\quants(S)$ but with estimates $\noisy{\delta}_{i_j}$. Let $\opti{S}$ be the optimal solution of the problem with small agents without noise. Since $(1 - 1/n^3) \delta_{i_j} \leq \noisy{\delta}_{i_j} \leq \delta_{i_j}$ for all $i,j$, we get that $\noisy{\quants}(\opti{S}) \geq (1 - 1 / n) \quants(\opti{S})$.  By the concavity of $\mlext(\cdot)$ along positive lines of direction, we get that $\mlext(\noisy{\quants}(\opti{S})) \geq (1 - 1 / n)\mlext(\quants(\opti{S}))$. 
\end{proof}

\begin{proof}[Proof of Lemma~\ref{l:ngreedyreduced}] 
First note that the objective function for the reduced instance is a submodular function regardless of the values of $\noisy{\delta}_{i_j}$. So since we are comparing ourselves with $\noisy{\quants}(S^*)$,
 it remains to show that the greedy algorithm with a noisy oracle on marginal contribution of agents performs well. 

Let $g(\cdot)$ be the objective function of the reduced instance. The marginal contributions are estimated by taking $\frac{10}{\delta^4}(1 + \ln n)$ samples of the random set with independent marginal probabilities $\quants$. By using basic Chernoff bounds as in \citet{CCMV-11}, we get that with high probability, all the estimates that are computed during the algorithm  have an additive error of at most $\delta^2 \redval(\opti{S})$.

Let $S$ be the set of small agents returned by the algorithm. Let $S_i = \{e_1, \dots, e_i\}$ be the value of $S$ after $i$ iterations. Now since $\redval(\cdot)$ is submodular, 
\begin{align*}
\redval(\opti{S}) & \leq \redval(S_{i-1}) + \sum_{e \in \opti{S} \setminus S_{i-1}} \redval_{S_{i-1}}(e) 
\end{align*}
By the greediness of the algorithm, $\noisy{\redval}_{S_{i-1}}(e_i) \geq \noisy{\redval}_{S_{i-1}}(e)$ for all $e \in \opti{S} \setminus S_{i-1}$. So, $\redval_{S_{i-1}}(e_i)  + 2 \delta^2 \redval(\opti{S}) \geq \redval_{S_{i-1}}(e)$, and 
\begin{align*}
\redval(\opti{S}) & \leq \redval(S_{i-1}) + \frac{1}{\delta} (\redval_{S_{i-1}}(e_i)  + 2 \delta^2 \redval(\opti{S}) ) \\
(1 - 2 \delta) \redval(\opti{S}) & \leq \redval(S_{i-1}) + \frac{1}{\delta} \redval_{S_{i-1}}(e_i)   
\end{align*}
Then, by following identically the remaining of the proof for the $e/ (e-1)$ approximation for greedy subject to a cardinality constraint, but by replacing $\redval(\opti{S})$ by $(1 - 2 \delta) \redval(\opti{S})$, we get that $(1 - 1/e) \redval(S) \geq (1 - 2 \delta) \redval(\opti{S}) $, which concludes the proof.
\end{proof}

\end{document}